\documentclass[twocolumn]{autart}    % Enable this line and disable the 
\usepackage{graphicx}   
\usepackage[all]{xy}
\usepackage{mathrsfs,array}

\usepackage[round]{natbib}
\usepackage{wrapfig}
\usepackage{algorithmic}
\usepackage{amsmath,amssymb,amsfonts}
\usepackage{float}
\usepackage{booktabs}
\usepackage{ntheorem}
\newtheorem*{proof}{Proof}
\newtheorem{theorem}{Theorem}[section]

\newtheorem{lemma}[theorem]{Lemma}

\newtheorem{example}[theorem]{Example}
\newtheorem{definition}[theorem]{Definition}
\newtheorem{remark}[theorem]{Remark}
\makeatletter

\newcommand{\Rmnum}[1]{\expandafter\@slowromancap\romannumeral #1@}
\makeatother
\begin{document}

\begin{frontmatter}

\title{Controlling the occurrence sequence of reaction modules through biochemical relaxation oscillators\thanksref{footnoteinfo}} % Title, preferably not more 

\thanks[footnoteinfo]{This paper is an extension of our earlier paper in the 22nd IFAC World Congress, July 9-14, 2023, Yokohama, Japan. Corresponding author C. Gao. Tel. +86-571-87952431. 
Fax +86-571-87953794.}

\author[China]{Xiaopeng Shi}\ead{xpshi@zju.edu.cn},    % Add the 
\author[China]{Chuanhou Gao}\ead{gaochou@zju.edu.cn},               % e-mail address 
\author[Belgium]{Denis Dochain}\ead{denis.dochain@uclouvain.be}  % (ead) as shown

\address[China]{School of Mathematical Sciences, Zhejiang University, Hangzhou, China}  % Please supply                                                        % full addresses
\address[Belgium]{ICTEAM, UCLouvain, B\^{a}timent Euler, avenue Georges Lema\^{i}tre 4-6, 1348 Louvain-la-Neuve, Belgium}        % here.

                             % keyword list or with the 
                                          % help of the Automatica 
                                          % keyword wizard

\begin{abstract}                          % Abstract of not more than 200 words.
Embedding sequential computations in biochemical environments is challenging because the computations are carried out by chemical reactions, which are inherently disordered. In this paper we apply modular design to specific calculations through chemical reactions and provide a design scheme of biochemical oscillator models in order to generate periodical species for the order regulation of these reaction modules. We take the case of arbitrary multi-module regulation into consideration, analyze the main errors in the regulation process under \textit{mass-action kinetics} and demonstrate our design scheme under existing synthetic biochemical oscillator models.
\end{abstract}

\end{frontmatter}

\section{Introduction}
Molecular computation, to be exact, embedding computational instructions into biochemical environments and executing them on behalf of molecules, has always been an important and difficult field in synthetic biology \citep{Vasic2020}. The approach is usually replacing the variables in the calculation instruction with the concentration of specific biomolecules, so that the change of concentrations could achieve a given calculation function according to well designed evolutionary rules. However, the biochemical environment has its complexity, so modular design based on biochemical reactions is not just a one-to-one copy of similar engineering structures and computer instructions \citep{Del2016}. A typical problem is to separate these special computational reaction modules with as little human intervention as possible, so that the otherwise disordered reactions can automatically occur, or just ``execute'' along the given sequence. Unlike the design of cell-like compartments and insertion of electrical clock signals \citep{Blount2017}, our strategy for the so-called module regulation is to construct biochemical oscillators which generate periodical components substituting the electrical signals to control the occurrence sequence of reaction modules.

 The research and design on biochemical oscillators are commonly with two main routes, one is experimental observation based on real biochemical structure and data \citep{McMillen2002, Montagne2011}, and the other is qualitative analysis based on abstract chemical reaction (or just chemical reaction network theory) and dynamical models \citep{Banaji2018, Angeli2021}. Considering that real biochemical experiments often only observe and estimate the existence of oscillation phenomena, while the means based on dynamic modeling and analysis can obtain richer understanding on the oscillation architectures and properties, we take the latter route to explore suitable oscillator models for our control tasks. In our previous work, we designed a relaxation oscillator model based on Chemical reaction network (CRN for short) for two-module regulations \citep{Shi2022} and three-module regulations \citep{Shi2023}. However, the regulation scheme for arbitrary multi-module regulations has not been well given, which is the main point of this paper. We also take the error analysis into consideration, since the low amplitude of our periodical signals can not stay strict $0$ which would lead to an inability to fully ``turn off'' the corresponding modules.

In this paper, we continue to consider the multi-module regulation based on the relaxation oscillation structure, the main contributions are listed below:
\begin{itemize}
    \item we give a design scheme for the regulation of arbitrary multiple modules, which can equally allocate the execution time for every module and therefore has better effects than the one in our conference paper \citep{Shi2023}.
    \item We analyze the errors introduced by the periodic signals.
    \item The strategy of loop termination is embedded in the scheme of sequence regulation, which enables the automatic computation in biochemical environments.
    \item We demonstrate our design scheme under a synthetic genetic relaxation oscillator and the Oregonator model separately given by \cite{McMillen2002} and \cite{Field1972}, showing that our oscillator design provides guidance for module regulation tasks in practical biochemical experiments.
\end{itemize}

This paper is organized as follows. The preliminaries about chemical reaction network is shown in Section \ref{sec:2}. The introduction on regulation mechanism and feasibility analysis of the design scheme for arbitrary multi-module regulation are presented in Section \ref{sec:3}. In Section \ref{sec:4} we give the analysis of the errors induced by the periodical signals and demonstrate the simulation result in a four-module regulation example. Section \ref{sec:5}
is dedicated to the discussion of termination strategy and biochemical application. And the last Section is the conclusion.

\textbf{Notations:} A capital letter (resp., $V$) usually refers to a particular chemical component and the corresponding lowercase letter (resp., $v$) represents its concentration. The sets of positive integers, non-negative integers, real numbers, non-negative real numbers and positive real numbers are denoted by $\mathbb{Z}_{>0}, \mathbb{Z}_{\geq 0}, \mathbb{R}, \mathbb{R}_{\geq 0}$ and $\mathbb{R}_{>0}$, respectively. We use $\mathbb{R}^n$ to denote an $n$-dimensional Euclidean space, and a vector $\boldsymbol{s} \in \mathbb{R}^n_{\geq0}$ (resp., $\boldsymbol{s} \in \mathbb{R}^n_{>0}$) if each element $s_{i} \in \mathbb{R}_{\geq 0}$ (resp., $s_{i} \in \mathbb{R}_{> 0}$) with $i=1,2,...n$.
% OR

%\begin{figure}
%\begin{center}
%\epsfig{file=jcaesar,width=7cm}
%\caption{Gaius Julius Caesar, 100--44 B.C.}
%\label{fig1}
%\end{center}
%\end{figure}

\section{Preliminaries on CRN}\label{sec:2}
A CRN is often expressed as $(\mathcal{S}, \mathcal{C}, \mathcal{R})$, with the three sets representing species set, complex set and reaction set, separately \citep{Feinberg2019}.
To be specific, for a CRN with $n$ species and $m$ reactions, usually the $j$th reaction is described in the following form:
\begin{equation}
R_j:~~ a_{j1}S_{1}+\cdots +a_{jn}S_{n} {\rightarrow}  b_{j1}S_{1}+\cdots +b_{jn}S_{n}.
\end{equation} 
The parameters $a_{ji},~ b_{ji} \in \mathbb{Z}_{\geq 0}$ refer to the stoichiometric constants of species $S_i$ as reactant and product. Then the three sets of a universal CRN can be denoted as $\mathcal{S}=\left \{S_1, S_2, \dots, S_n \right \}$, $\mathcal{C}= \left({\textstyle \bigcup_{j=1}^{m}}\sum_{i=1}^{n} a_{ji}S_i\right)\cup   \left({\textstyle \bigcup_{j=1}^{m}}\sum_{i=1}^{n} b_{ji}S_i\right)$ and $\mathcal{R}=\left \{R_1, R_2, \dots, R_m \right \}$. The corresponding dynamics describing the evolution of species concentrations can be written as 
\begin{equation}\label{dynamics1}
    \dot{\boldsymbol{s}} = \Xi \cdot \boldsymbol{r}\left(\boldsymbol{s}\right)\ ,
\end{equation}
where $\boldsymbol{s}=\boldsymbol{s}(t) \in \mathbb{R}^n_{\geq 0}$ represents the concentration vector of $\mathcal{S}$, matrix $\Xi_{n \times m}$ is the stoichiometric matrix defined by $\Xi_{ji}=b_{ji}-a_{ji}$, and $\boldsymbol{r}(\boldsymbol{s}) \in \mathbb{R}^m$ is the rate vector determined by the specific kinetics assumption. This paper takes the \textit{mass-action kinetics} as the modeling assumption which has the following expression:
\begin{equation}\label{dynamics2}
\boldsymbol{r}(\boldsymbol{s}) = \left ( k_{1}\prod_{i=1}^{n}s_{i}^{a_{1i}},\dots,k_{m}\prod_{i=1}^{n}s_{i}^{a_{mi}}  \right )^{\top}
\end{equation}
with the reaction rate constant of $R_j$ to be $k_j >0$. We often omit the rate constant in the context if $k_j=1$. A CRN equipped with \textit{mass-action kinetics} is called a mass-action chemical system, which can be modeled as a group of polynomial ODEs. It has been proved that mass-action chemical kinetics is Turing universal, which means that any functional computation can be embedded into solutions of a group of polynomial ODEs and then be implemented back into mass-action chemical systems \citep{Fages2017, Salehi2017}. Our description on the reaction modules to be regulated and the design of the oscillator are both based on this.

\section{Mechanism and universal approach for arbitrary multi-module regulations}\label{sec:3}
In this section we provide introduction on our basic oscillator model, and then demonstrate the design scheme for arbitrary multi-module regulations.

\subsection{Regulation mechanism of relaxation oscillator}
We begin with a two-module example to show the mechanism of module regulation.

\begin{example}\label{two:module}
Consider two reaction modules $\mathcal{M}_1$ and $\mathcal{M}_2$, given by
\begin{align*}
\mathcal{M}_1:
        S_{1} &\to S_{1} + S_{2}\ , &\mathcal{M}_2:  S_{2} &\to S_{1} + S_{2}\ , \\
        S_{3} &\to S_{3} + S_{2}\ ,  &   S_{1} &\to \varnothing\ . \\
        S_{2} &\to \varnothing \ ; 
    \end{align*}    
\end{example}
The reaction modules to be regulated can be viewed as CRNs which realize specific computations like the computer instructions based on the equilibria of the network dynamics. Note that species $S_3$ is a catalyst and we treat the equilibrium concentration $s_1^*$ and $s_2^*$ of $S_1$ and $S_2$ respectively as the outputs of the modules, then it is obvious that the two modules achieve instructions like $s_2^*=s_1^*+s_3^*$ and $s_1^*=s_2^*$ based on the ODEs as follows:
\begin{equation}\label{ode:ex1}
\begin{aligned}
\mathcal{M}_1:
\dot{s}_2 &= s_{1} + s_{3} - s_{2}\ , & 
\mathcal{M}_2:\dot{s}_1 &= s_{2} - s_{1}\ , \\
\dot{s}_1 &=\dot{s}_3 = 0\ ;  &  \dot{s}_2 &= 0\ .  
\end{aligned}
\end{equation}
After we set the initial concentration of $S_3$ as $1$, alternative loop of these two modules can realize the instruction as $s_1^*=s_1^*+1$, just like the increment operation. Alternation and loop of these two modules call for suitable periodical clock signals, so our work is to design appropriate oscillator model to generate species acting as these signals.

The approach is based on the relaxation oscillation and truncated subtraction operation, which can be expressed in the following ODEs:
\begin{subequations}\label{basic odes}
    \begin{align}
        \epsilon_{1}\dot{x} &= \eta_{1}f(x,y)   \ , \label{eq1a} \\
         \dot{y} &= \eta_{1}g(x,y)  \ , \label{eq1b} \\
    \epsilon_{1}\epsilon_{2}\dot{u} &= \eta_{1}(\epsilon_{1}(p-u)-uv)  \ , \label{eq1c} \\
    \epsilon_{1}\epsilon_{2}\dot{v} &= \eta_{1}(\epsilon_{1}(x-v)-uv)  \ , \label{eq1d}
    \end{align}
\end{subequations}
with $0 < \epsilon_{1}, \epsilon_{2} \ll 1$ and $\eta_{1}, p >0$. 
The parameter $\epsilon_1$ leads to the relaxation oscillation in the system together with the function $f(x,y)$ and $g(x,y)$, $\epsilon_2$ ensures that $u$ and $v$ oscillate synchronously with the driving element $x$, and $\eta_1$ controls the period of the whole system. We have proved that the low-amplitude segments of $u$ and $v$ can be close enough to $0$ under certain conditions of parameters and initial values, and demonstrated the process of two-module regulation in \citep{Shi2022} based on the detailed ODEs:
\begin{subequations}\label{detailed odes}
    \begin{align}
         \Sigma_{xy}: \epsilon_{1} \dot{x} &=\eta_{1}(\varphi(x)-y)x \ , \label{eq2a} \\
         \dot{y} &=\eta_{1}(x-\ell)y \ , \label{eq2b} \\
         \Sigma_{uv}: \epsilon_{1}\epsilon_{2}\dot{u} &= \eta_{1}(\epsilon_{1}(p-u)-uv) \ , \label{eq2c} \\
    \epsilon_{1}\epsilon_{2}\dot{v} &= \eta_{1}(\epsilon_{1}(x-v)-uv)\ , \label{eq2d}
    \end{align}
\end{subequations}
with $\varphi(x)=-x^3+9x^2-24x+21$. The subsystem $\Sigma_{xy}$ has a cubic-form critical manifold as $\left\{(x,y):y=\varphi(x) \right \}$ which has two fold points $x=2$ and $x=4$ (Referring to Fig. \ref{fig1}). We also command that $2<\ell, p<4$ so that the relaxation oscillation emerges in $\Sigma_{xy}$ and is inherited in a specific form by $\Sigma_{uv}$. This model can tackle the regulation task in Example \ref{two:module} as long as we insert the pair of symmetrical clock signals $U$ and $V$ separately into the $\mathcal{M}_1$ and $\mathcal{M}_2$ as catalysts, achieving the following modules:
\begin{equation*}
\begin{array}{c|c}
\underbrace{\begin{aligned}
S_{1} + U& \to S_{1} + S_{2} + U\ , \\
 S_{3} + U&\to S_{3} + S_{2} + U\ , \\
S_{2} + U&\to U\ ;
\end{aligned}}_{\tilde{\mathcal{M}}_{1}}
&
\underbrace{\begin{aligned}
S_{2} + V&\to S_{1} + S_{2} + V\ , \\
 S_{1} + V&\to V\ .
\end{aligned}}_{\tilde{\mathcal{M}}_{2}}
\end{array}
\end{equation*}
Signals $U$ and $V$ oscillate in an anti-phase form, so the reactions in $\tilde{\mathcal{M}}_{1}$ and $\tilde{\mathcal{M}}_{2}$ will occur in sequence at staggered intervals, even if they are mixed together. However, the design of \eqref{detailed odes} can not directly handle multi-module regulations. We follow the utilization of periodical components as catalysts and provide a further scheme in this paper to generate arbitrary number of symmetrical clock signals.

Before the introduction of our scheme, we need to extend the requirements for a pair of symmetrical clock signals to arbitrary $m$ clock signals in the Definition \ref{def1}. 
\begin{definition}[symmetrical clock signals]\label{def1}
    Species $V_1, V_2, \dots, V_m$ are called symmetrical clock signals ($m \in \mathbb{Z}_{>0}$) if 
    \begin{enumerate}
        \item [\rm 1.] $V_1, V_2, \dots, V_m$ oscillate synchronously with the same period $T$ which can be equally divided into $m$ separate parts.
        \item [\rm 2.] Only the signal $V_j$ has a concentration strictly larger than $0$ (e.g. $v_j>1$) in the $j$th segment of each single period $T$ (i.e. $\frac{T}{m}(j-1)<t<\frac{T}{m}j$), and concentrations of other signals are in a tiny neighborhood of $0$ (i.e. $v_i \ll 1$ when $i \ne j$). 
        \item [\rm 3.] The concentration of each signal switches between the high and low amplitude in the form of abrupt transitions between phases.
    \end{enumerate}
\end{definition}
Note that a single driving component $X$ can only generate a pair of anti-phase signals $U$ and $V$ which can not simultaneously act as two elements in the $m$ clock signals, so we have to reuse the system \eqref{detailed odes} under different $m$ initial conditions and take a separate result of $V_j$ each time as the clock signal ($j=1,2, \dots, m$).
Fig. \ref{fig2} provides an illustration with $m=4$, which is actually the simulation result we will describe in the following four-module regulation situation.

Imaging that there are four modules named $\mathcal{M}_1$, $\mathcal{M}_2$, $\mathcal{M}_3$ and $\mathcal{M}_4$ to be regulated, we have to produce corresponding four symmetrical signals $V_1$, $V_2$, $V_3$ and $V_4$ satisfying our requirements in Definition \ref{def1}. We set $\ell=2.3158$ in order to let the period ($T_1$) corresponding to the low amplitude of $X$ is three times as much as the high one ($T_2$), that is to say,
\begin{equation}\label{ratio}
    \frac{T_1}{T_2} \approx \frac{\int_{1}^{2}\frac{\varphi'(x)dx}{\eta_{1}(x-\ell)\varphi(x)}}{\int_{5}^{4}\frac{\varphi'(x)dx}{\eta_{1}(x-\ell)\varphi(x)}}=3,
\end{equation}
and then mark out four points $P_1$, $P_2$, $P_3$ and $P_4$ on the critical manifold of $\Sigma_{xy}$ whose coordinates correspond to the following formula:
\begin{equation}
    \frac{\int_{x_i}^{2}\frac{\varphi'(x)dx}{\eta_{1}(x-\ell)\varphi(x)}}{\int_{1}^{2}\frac{\varphi'(x)dx}{\eta_{1}(x-\ell)\varphi(x)}}=\frac{i-1}{3}, ~~  y_i=\varphi(x_i), ~~ i=1,2,3,4.
\end{equation}
Fig. \ref{fig1} shows the diagram for the phase plane of the subsystem $\Sigma_{xy}$ and the general positions of the four initial points can be found as $x_1=2, x_2=1.6463, x_3=1.3776, x_4=1$. The oscillation orbit $\Gamma_{\epsilon_1}$ consists of two segments which lie in the $O(\epsilon_1)$-neighborhood of the critical manifold and two horizontal segments. With $\epsilon_{1}=\epsilon_{2}=0.001, \eta_{1}=0.1, p=3$ and $P_i$ with $(u_i,v_i)=(0,0)$ as the initial points, simulation result for $v_i$ is provided in Fig. \ref{fig2}. The rules for $P_i$ shift the initial phase of driving component $x$ and correspondingly stagger the phases of $v_i$, making sure that there is only one $v_i$ strictly larger than $0$ in every segment of the whole period $T$. 

\begin{figure}
\begin{center}
\includegraphics[width=1.0\linewidth,scale=1.00]{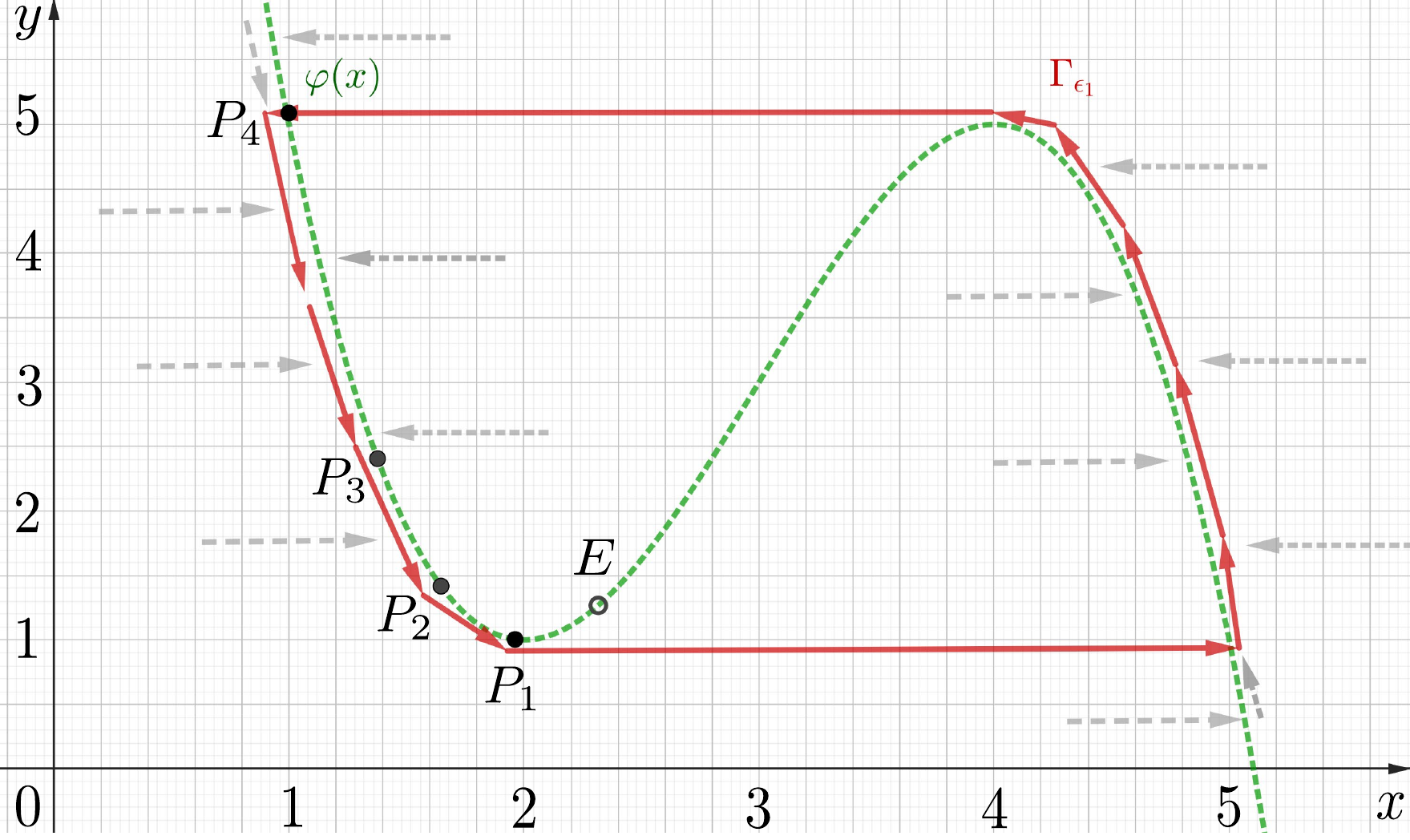}    % The printed column  
\caption{Diagram for the phase plane of $\Sigma_{xy}$. The dotted green curve represents the critical manifold $\left \{ (x,y): y=\varphi(x) \right \}$ and the red lines describe the approximate position of the oscillation orbit $\Gamma_{\epsilon_1}$. $P_1, P_2, P_3$ and $P_4$ are the four initial points, and the point $E(\ell,\varphi(\ell))$ lying on the middle part of critical manifold is the unique unstable equilibrium in the first quadrant.}  % width is 8.4 cm.
\label{fig1}                                 % Size the figures 
\end{center}                                 % accordingly.
\end{figure}

\begin{figure}
\begin{center}
\includegraphics[width=1.0\linewidth,scale=1.00]{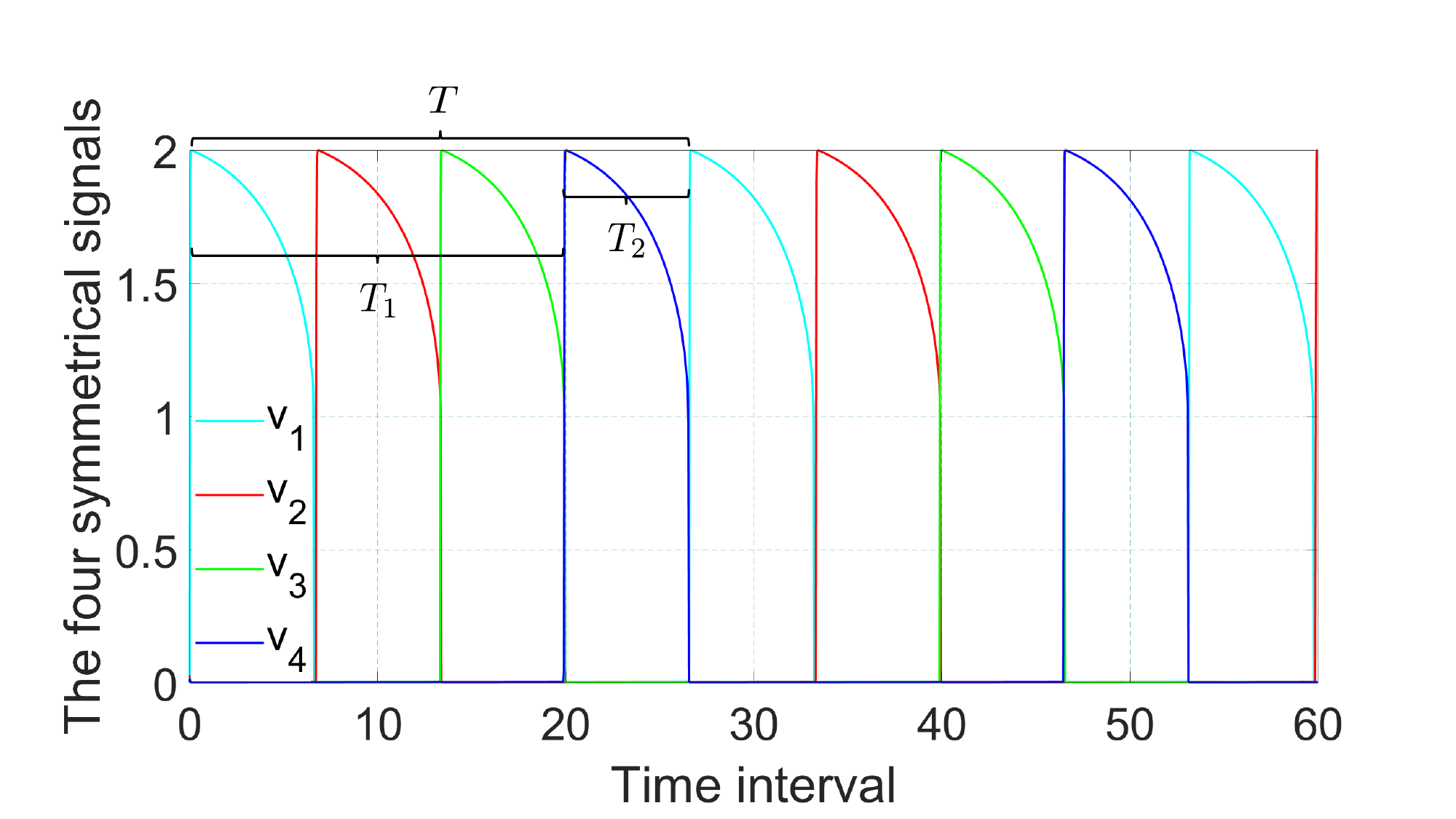}    % The printed column  
\caption{Simulation result for the concentrations of the four signals $V_1$, $V_2$, $V_3$ and $V_4$. The four signals are symmetrical as Definition \ref{def1} with $T \approx 26.874$. }  % width is 8.4 cm.
\label{fig2}                                 % Size the figures 
\end{center}                                 % accordingly.
\end{figure}

It should be noted that the important premise for the construction of the above four signals is the formula \eqref{ratio}, so the ratio of $T_1$ and $T_2$ actually determines the number of modules which can be regulated. For a regulation task of $m$ modules ($m \geq 3$), the corresponding ratio should be $m-1$. In our previous model \eqref{detailed odes}, only the parameter $\ell$ adjusts the ratio, that is $0.256 < T_1/T_2 < 6.987$ within $2< \ell <4$. Therefore, our original model design fails for regulation tasks of more than seven modules. The strategy here is to replace the detailed $\varphi(x)$ with general cubic functions, and the analysis for the parameter selection is provided in the following subsection.

\subsection{Design scheme for arbitrary multi-module regulations}
In order to explore the range of $T_1/T_2$, we extend the expression of $\varphi(x)$ to the general cubic form as 
\begin{equation}\label{eq:cubic}
\varphi(x)=-x^3+bx^2+cx+d,   
\end{equation}
and the ratio will inherit the form in formula \eqref{ratio}. The following lemma gives the brief constraints on coefficient values.
\begin{lemma}
    To ensure the existence of relaxation oscillation and the correspondence to mass-action CRN systems, parameters in (\ref{eq:cubic}) should satisfy the restrictions below:
    \begin{enumerate}
        \item [\rm 1.] $b>0$, $c<0$ and $-3c<b^2<-4c$;
        \item [\rm 2.] $d>d_0=(-b(b-\sqrt{b^2+3c})^2-6c(b-\sqrt{b^2+3c}))/27$;
        \item [\rm 3.] $x_m< \ell <x_M$, where $x_m$ and $x_M$ represent the two roots of the equation $\varphi^{'}(x)=-3x^2+2bx+c=0$.
    \end{enumerate}
\end{lemma}
\begin{proof}
    The restrictions directly come from the following understanding:
    (1). The critical manifold $y=\varphi(x)$ should be cubic-shaped with two fold points $x_m$ and $x_M$, which induces $b>0$, $c<0$. (2). We name the four vertices of the possible oscillation orbit $\Gamma_{\epsilon_1}$ as $(x_l, \varphi(x_M)), (x_m, \varphi(x_m)), (x_M, \varphi(x_M))$ and $(x_h, \varphi(x_m))$ with $x_l<x_m<x_M<x_h$, then the biochemical correspondence calls for the orbit $\Gamma_{\epsilon_1}$ lying strictly in the first quadrant of the phase plane, i.e., $x_l >0$ and $\varphi(x_m)>0$, which leads to $-3c<b^2<-4c$ and $d>d_0=(-b(b-\sqrt{b^2+3c})^2-6c(b-\sqrt{b^2+3c}))/27$. (3). The equilibrium point $E$ should emerge on the middle part of the critical manifold, that is to say, $x_m< \ell <x_M$. We provide a brief glance in Fig. \ref{fig3} and skip the detailed derivations.  $\hfill \blacksquare$
\end{proof}
Fig. \ref{fig3} also shows that $T_1$ (resp., $T_2$) refers to the time it takes to travel along the segment of oscillation orbit between $x_l$ and $x_m$ (resp., $x_h$ and $x_M$). We could omit the time spent on the two horizontal segments because it is an instantaneous behavior controlled by the $\epsilon_1$. The following lemma helps to simplify the analysis of $T_1/T_2$.

\begin{figure}
\begin{center}
\includegraphics[width=1.0\linewidth,scale=1.00]{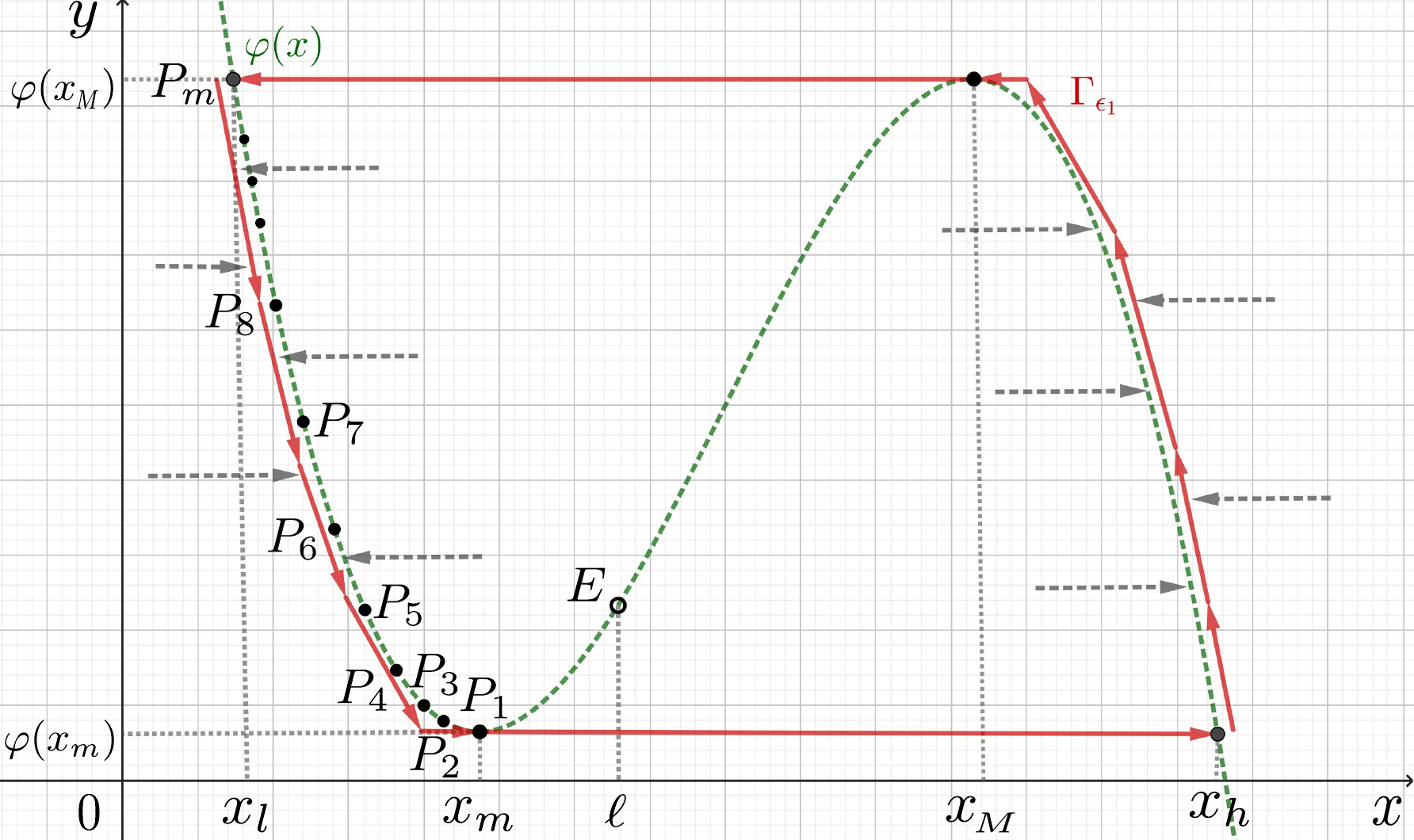}    % The printed column  
\caption{A brief glance on the phase plane under a general $\varphi(x)$ with the unique unstable equilibrium $E$ and the approximate position of the oscillation orbit $\Gamma_{\epsilon_1}$ described by the connected red lines for a regulation task of $m$ modules.}  % width is 8.4 cm.
\label{fig3}                                 % Size the figures 
\end{center}                                 % accordingly.
\end{figure}

\begin{lemma}
    The ratio $\frac{T_1}{T_2} \approx \frac{\int_{x_l}^{x_m}\frac{\varphi'(x)dx}{\eta_{1}(x-\ell)\varphi(x)}}{\int_{x_h}^{x_M}\frac{\varphi'(x)dx}{\eta_{1}(x-\ell)\varphi(x)}}$ decreases as $\ell$ increases with $x_m< \ell <x_M$.
\end{lemma}
\begin{proof}
    It is easy to check that $\frac{\partial T_1}{\partial \ell} <0$ and $\frac{\partial T_2}{\partial \ell} >0$ with $x_m< \ell <x_M$, so the lemma is obvious. $\hfill \blacksquare$
\end{proof}
 To explore the upper bound of the ratio $T_1/T_2$, we can let $\ell$ be close enough to $x_m$ and consider an alternative form as
\begin{equation}
    \frac{T_1}{T_2} \approx R(b,c,d)=\frac{\int_{x_l}^{x_m}\frac{x_M-x}{-x^3+bx^2+cx+d}dx}{\int_{x_h}^{x_M}\frac{x_M-x}{-x^3+bx^2+cx+d}dx}.
\end{equation}
Specifically, we focus on the boundary cases that $b^2=-4c$ and $d=d_0$, then the expression of $R(b,c,d)$ can be further simplified as 
\begin{equation}\label{R(b)}
    R(b)=\frac{\int_{0}^{\frac{1}{6}b}\frac{\frac{1}{2}b-x}{-(x-\frac{1}{6}b)^2(x-\frac{2}{3}b)}dx}{\int_{\frac{2}{3}b}^{\frac{1}{2}b}\frac{\frac{1}{2}b-x}{-(x-\frac{1}{6}b)^2(x-\frac{2}{3}b)}dx}.
\end{equation}
Noting that $b^2=-4c$ and $d=d_0$ separately contributes to $x_l=0$ and $\varphi(x_m)=0$, the former has no impact on the subsequent analysis, while the latter one leads to the disappearance of the oscillation (this is because the possible oscillation orbit intersects the X-axis, making all the points on the axis become degenerate stable equilibria). Based on these consideration, we need to introduce a little disturbance to avoid the case $d=d_0$, which inspires us to tackle with the following form 
\begin{equation}\label{specific form}
    R(b,\alpha,\beta)=\frac{\int_{0}^{\frac{1}{6}b}\frac{\frac{1}{2}b-x}{-((x-\frac{1}{6}b)^2+\alpha)(x-\frac{2}{3}b-\beta)}dx}{\int_{\frac{2}{3}b}^{\frac{1}{2}b}\frac{\frac{1}{2}b-x}{-((x-\frac{1}{6}b)^2+\alpha)(x-\frac{2}{3}b-\beta)}dx},
\end{equation}
with $0<\alpha, \beta \ll 1$ instead of directly substituting $d>d_0$ for easier analysis. The introduction of $\alpha$ and $\beta$ leads to the change of overall structure of $\varphi(x)$, which should be regarded as the disturbance on the expression of \eqref{R(b)}. We actually utilize $R(b, \alpha, \beta)$ in \eqref{specific form} under sufficiently small $\alpha$ and $\beta$ to approximate the possible values of $R(b)$ in \eqref{R(b)}. The following theorem shows that $R(b,\alpha,\beta)$ can go to the infinite, whose proof can be found in Appendix.

\begin{theorem}\label{thm1}
$\lim_{\alpha,\beta \to 0}R(b,\alpha ,\beta)=\infty$, $\forall b>0$ if $\lim_{\alpha,\beta \to 0}\frac{\beta}{\alpha} \ne 0$. 
\end{theorem}

Based on the theorem, $R(b,c,d)$, as the approximation of $T_1/T_2$, could also be large enough under the boundary condition $c=-\frac{1}{4}b^2,\  d \to d^+_0$ (i.e., $d-d_0$ converges to $0$ from the positive side). We then provide some simulation data based on the values of $b,c, d$ and $\ell$ in the Table \ref{tab:my_label}, showing that the range of $T_1/T_2$ is more than enough for the scale of arbitrary module regulation tasks in practice. 

\begin{table}[H]
    \caption{Simulation of $T_1/T_2$ under different choice of parameters.}
    \label{tab:my_label}
\begin{center}
    \begin{tabular}{cccccc}
    \toprule
        $b$ & $c$ & $d_0$ & $d$ & $\ell$ & $T_1/T_2$ \\
        \midrule
        12 & -36 & 32 & 32.1 & 2.2 & 9.319 \\
        12 & -36 & 32 & 32.1 & 2 & 23.976 \\
        12 & -36 & 32 & 32.01 & 2 & 55.547 \\
        16 & -64 & 2048/27 & 2048/27+0.01 & 8/3 & 77.582 \\
        20 & -100 & 4000/37 & 4000/27+0.01 & 10/3 & 101.106 \\
        20 & -100 & 4000/37 & 4000/27+0.001 & 10/3 & 259.200 \\
        24 & -144 & 256 & 256.001 & 4 & 326.026 \\
        30 & -225 & 500 & 500.0001 & 5 & 1166.990 \\
        \bottomrule
    \end{tabular}
    \end{center}
\end{table}
\begin{remark}
    One may note that although the boundary condition $c=-\frac{1}{4}b^2$ does not interfere the analysis of the oscillation, it would destroy the biochemical implementation because the actual oscillation orbit would go out of the first quadrant when $x_l=0$. The remedy is to let the parameter $c$ slightly smaller than $-\frac{1}{4}b^2$, which has little influence on the estimation of $T_1/T_2$. 
\end{remark}

\section{Error analysis and control induced by the clock signals}\label{sec:4}
In this section, we provide the analysis of calculation errors introduced by our symmetrical signals, and then discuss the corresponding measure to reduce the errors.

For a regulation task of $m$ modules $\mathcal{M}_1, \mathcal{M}_2, \dots, \mathcal{M}_m$ and $n$ species $S_1, S_2, \dots S_n$, our strategy and analysis are based on the specific form of system \eqref{detailed odes} shown in the previous section. An appropriate set of $(b,c,d,\ell)$ makes sure $T_1/T_2 = m-1$, with the period for a single loop of these $m$ modules to be $T=T_1+T_2$. According to the following equation:
\begin{equation}\label{eq:x_i}
    \frac{\int_{x_j}^{x_m}\frac{\varphi'(x)dx}{\eta_{1}(x-\ell)\varphi(x)}}{\int_{x_l}^{x_m}\frac{\varphi'(x)dx}{\eta_{1}(x-\ell)\varphi(x)}}=\frac{j-1}{m-1},  ~~ j=2, \dots, m-1,
\end{equation}
the $m$ initial points $P_1=(x_m,\varphi(x_m)), P_j=(x_j,\varphi(x_j))$, $P_m=(x_l,\varphi(x_M))$ with $j=2, \dots, m-1$ in the phase plane of $\Sigma_{xy}$ (see Fig. \ref{fig3}) can finally induce $m$ symmetric oscillatory signals $V_1, V_2, \dots, V_m$. Logically, every oscillatory signal $V_j$ enters into the corresponding module $\mathcal{M}_j$ and participates all reactions within $\mathcal{M}_j$ as a catalyst, e.g., the reaction $A\to B$ in $\mathcal{M}_j$ will change to be $A+V_j\to B+V_j$. Therefore, the oscillatory signal $V_j$ can control the occurrence or closure of $\mathcal{M}_j$ according to the nonzero or zero concentration of $V_j$. Moreover, the dynamics of the original species in $\mathcal{M}_j$ will keep unchanged whatever the concentration of $V_j$ is. However, it should be pointed out that the low amplitude of $v_j$ can not stay exactly at $0$ and the error is controlled by $O(\epsilon_1)$ as the following lemma says:
\begin{lemma}[\cite{Shi2022}]
    The low amplitude of $v_j$ can be viewed as a function of the driving element $x$ under the general cubic form of $\varphi(x)$ in (\ref{detailed odes}) as 
    \begin{equation}
        |v_j-0|< \frac{x_m}{p-x_m}\epsilon_1 = O(\epsilon_1),
    \end{equation}
    which corresponds to $x$ at the low amplitude (i.e., $x_l<x<x_m$) for $\forall j=1, 2, \dots, m$.
\end{lemma}

After inserting the clock signal $V_j$ into $\mathcal{M}_j$ as catalyst ($j=1, 2, \dots, m$), the execution time allocated for each single module is exactly $T/m$. We also name the time segment when $\mathcal{M}_j$ is turned on during the $k$th loop as $T^j_k$, which can be described as $$T^j_k= [(k-1)T+\frac{T}{m}(j-1), (k-1)T+\frac{T}{m}j].$$ Towards the $m$ modules to be regulated, the reactions in each module actually correspond to some specific arithmetic instructions. We claim that these reactions should converge exponentially to the stable equilibria. Under the {\it mass-action kinetics}, we can model the ODEs of the whole system as follows:
\begin{equation}\label{regulation system}
    \boldsymbol{\dot{s}}=\boldsymbol{H}(\boldsymbol{s}) \cdot \boldsymbol{v}.
\end{equation}
The vector $\boldsymbol{s}=(s_1, s_2, \dots, s_n)^{\top}$ and $\boldsymbol{v}=(v_1, v_2, \dots, v_m)^{\top}$ separately represents the concentration of the species and clock signals. The matrix $\boldsymbol{H}_{n \times m}(\boldsymbol{s})=(h_{i,j}(\boldsymbol{s}))$ describes the kinetics of the modules without clock signals and specifically, $h_{i,j}(\boldsymbol{s})$ refers to the kinetics of species $S_i$ in the module $\mathcal{M}_j$ in a polynomial form. We remind that $\frac{\partial h_{i,j}(\boldsymbol{s})}{\partial s_i} \leq 0$ and state it in the following lemma:
\begin{lemma}\label{lemma:4.2}
    The requirement for the existence of stable equilibrium in each $\mathcal{M}_j$ implies that $\frac{\partial h_{i,j}(\boldsymbol{s})}{\partial s_i} \leq 0$ for $\forall i=1, 2, \dots, n$ and $j=1, 2, \dots, m$. 
\end{lemma}
\begin{proof}
    We prove this lemma by proving its contrapositive. Assume that $\exists$ a set $(i,j)$ such that $\frac{\partial h_{i,j}(\boldsymbol{s})}{\partial s_i} > 0$ and the equilibrium of reactions in $\mathcal{M}_j$ can be noted as $(s_1^*, \dots, s_n^*)$. If we fix the value of $s_k=s_k^*$ with $k \ne i$ and increase $s_i$, then $\dot{s}_i=h_{i,j}(\boldsymbol{s})>0$ because $\frac{\partial h_{i,j}(\boldsymbol{s})}{\partial t}=\frac{\partial h_{i,j}(\boldsymbol{s})}{\partial s_i} \cdot \frac{\partial s_i}{\partial t} > 0$, which means $s_i$ will keep deviating from $s_i^*$ and the equilibrium becomes unstable. Therefore, the stable equilibrium of reactions in each $\mathcal{M}_j$ leads to nonpositive $\frac{\partial h_{i,j}(\boldsymbol{s})}{\partial s_i}$. 
    $\hfill \blacksquare$
\end{proof}

One can review the ODEs in \eqref{ode:ex1} of Example \ref{two:module} to verify 
Lemma \ref{lemma:4.2}. 

\begin{example}
 Our design of the reactions in $\mathcal{M}_1$ and $\mathcal{M}_2$ induces
 $H(s_1, s_2, s_3)=\begin{pmatrix}
  0&s_2-s_1 \\
  s_1+s_3-s_2&0 \\
  0&0
\end{pmatrix}$
satisfying $\frac{\partial h_{1,1}}{\partial s_1}=\frac{\partial h_{3,1}}{\partial s_3}=\frac{\partial h_{2,2}}{\partial s_2}=\frac{\partial h_{3,2}}{\partial s_3}=0$ and $\frac{\partial h_{2,1}}{\partial s_2}=\frac{\partial h_{1,2}}{\partial s_1}=-1$, which is consistent with Lemma \ref{lemma:4.2}.
\end{example}
During the segment $T^j_k$, we rewrite the signal vector $\boldsymbol{v}$ as $\boldsymbol{v}^j=(v^j_1, v^j_2, \dots, v^j_m)$, then except for $v^j_j$, the remaining $m-1$ elements are sufficiently small which can be controlled by $\epsilon=O(\epsilon_1)$. Since $v^j_j > x_M-p$, we can always keep $v^j_j \geq 1$ through selection of $x_M$ and $p$ so that this signal would accelerate the corresponding execution in $\mathcal{M}_j$ when it turns on the module.

The ideal regulation is under the condition $\epsilon=0$, so we consider the situation $\epsilon < \epsilon_0 \ll 1$, then the system \eqref{regulation system} becomes a singular perturbation problem. Focusing on the segment $T^j_k$ when the $\mathcal{M}_j$ is turned on while the others are off, we define a set $\mathcal{S}^j$ as a subset of the species set $\mathcal{S}$ that $\forall S_i \in \mathcal{S}^j, h_{i,j}(\boldsymbol{s}) \ne 0$. That is to say, $\mathcal{S}^j$ consists of species whose concentrations change in $\mathcal{M}_j$. We set $|\mathcal{S}^j|=l_j$ and renumber the set $\mathcal{S}$ by putting the $l_j$ species which belong to $\mathcal{S}^j$ in the front, just as $(S_1,S_2,\dots, S_n)=(\mathcal{S}_{\boldsymbol{\omega}} | \mathcal{S}_{\boldsymbol{\overline{\omega}}})$ with species in $\mathcal{S}_{\boldsymbol{\overline{\omega}}}$ act as catalysis in $\mathcal{M}_j$. We use $\boldsymbol{\omega}$ and $\boldsymbol{\overline{\omega}}$ to separately represent the concentration vector of species set $\mathcal{S}_{\boldsymbol{\omega}}$ and $\mathcal{S}_{\boldsymbol{\overline{\omega}}}$, then the original system \eqref{regulation system} during the segment $T^j_k$ correspondingly changes into
\begin{subequations}\label{system:E}
\begin{align}
     \dot{\boldsymbol{\omega}}&=\begin{pmatrix}
 \dot{s}_1\\
 \vdots\\
\dot{s}_{l_j}
\end{pmatrix}=
\begin{pmatrix}
H_1(\boldsymbol{s}) \\
 \vdots\\
H_{l_j}(\boldsymbol{s})
\end{pmatrix} \cdot \epsilon +
\begin{pmatrix}
h_{1,j}(\boldsymbol{s}) \\
 \vdots\\
h_{l_j,j}(\boldsymbol{s})
\end{pmatrix} \cdot v^j_j  \label{system:E:a} \ , \\
     \dot{\boldsymbol{\overline{\omega}}}&=\begin{pmatrix}
 \dot{s}_{l_j+1}\\
 \vdots\\
\dot{s}_n
\end{pmatrix}=
\begin{pmatrix}
H_{l_j+1}(\boldsymbol{s}) \\
 \vdots\\
H_{n}(\boldsymbol{s})
\end{pmatrix} \cdot \epsilon \label{system:E:b} \ , 
\end{align}
\end{subequations}

and the reduced system under $\epsilon=0$ can be described as
\begin{subequations}\label{system:E_0}
\begin{align}
     \dot{\boldsymbol{\omega}}_0&=\begin{pmatrix}
 \dot{s}_{0,1}\\
 \vdots\\
\dot{s}_{0,l_j}
\end{pmatrix}=
\begin{pmatrix}
h_{1,j}(\boldsymbol{s}_0) \\
 \vdots\\
h_{l_j,j}(\boldsymbol{s}_0)
\end{pmatrix} \cdot v^j_j  \label{system:E_0:a} \ , \\
\dot{\boldsymbol{\overline{\omega}}}_0&=\begin{pmatrix}
 \dot{s}_{0,l_j+1}\\
 \vdots\\
\dot{s}_{0,n}
\end{pmatrix} =0 \label{system:E_0:b} \ ,
\end{align}
\end{subequations}
with $\boldsymbol{\omega}, \boldsymbol{\omega}_0 \in \mathbb{R}^{l_j}_{\geq 0}, ~ \boldsymbol{\overline{\omega}}, \boldsymbol{\overline{\omega}}_0 \in \mathbb{R}^{n-l_j}_{\geq 0}, ~ H_i(\boldsymbol{s})=\sum_{k\ne j}h_{i,k}(\boldsymbol{s})$ and $\boldsymbol{s}_0=(\boldsymbol{\omega}_0|\boldsymbol{\overline{\omega}}_0)$ referring to the corresponding concentrations of species in the reduced system. Noting that the reactions in $\mathcal{M}_j$ converge to the equilibrium exponentially, system \eqref{system:E_0:a} has the solution $\boldsymbol{\omega}_0(t)$ in an exponential form, which can be approximated by the equilibrium equations $\left \{\boldsymbol{\omega}_0: h_{1,j}(\boldsymbol{s}_0)= \cdots h_{l_j,j}(\boldsymbol{s}_0)= 0 \right \}$ at the terminal of this segment. So the solution of the reduced system \eqref{system:E_0} can be described as $\boldsymbol{s}_0=(\boldsymbol{\omega}_0(t)|\boldsymbol{\overline{\omega}}_0(0))$. The sufficiently small $\epsilon$ implies that the system \eqref{system:E} is a regular perturbation system, then one could expect the original solution $\boldsymbol{s}$ as perturbation of the $\boldsymbol{s}_0$, i.e., $\exists$ vector functions $\boldsymbol{P}(t) \in \mathbb{R}^{l_j}$ and $\boldsymbol{Q}(t) \in \mathbb{R}^{n-l_j}$ which are not related to $\epsilon$ and satisfy
\begin{equation}
    \boldsymbol{\omega}=\boldsymbol{\omega}_0+\boldsymbol{P}(t) \cdot \epsilon \ , \boldsymbol{\overline{\omega}}=\boldsymbol{\overline{\omega}}_0+\boldsymbol{Q}(t) \cdot \epsilon \ .
\end{equation}
Substituting this formula back into the original system \eqref{system:E}, $\boldsymbol{P}(t)$ and $\boldsymbol{Q}(t)$ read like:
\begin{subequations}
    \begin{align}
 \boldsymbol{P}(t)&=\begin{pmatrix}
 \int_{t_1}^{t}H_1(\boldsymbol{s})dr  \\
 \vdots\\
 \int_{t_1}^{t}H_{l_j}(\boldsymbol{s})dr
\end{pmatrix} +
\begin{pmatrix}
\int_{t_1}^{t}\frac{h_{1,j}(\boldsymbol{s})-h_{1,j}(\boldsymbol{s}_0)}{\epsilon}  v^j_j dr\\
 \vdots\\
\int_{t_1}^{t}\frac{h_{l_j,j}(\boldsymbol{s})-h_{l_j,j}(\boldsymbol{s}_0)}{\epsilon} v^j_j dr  
\end{pmatrix} \label{P(t)} \ , \\
\boldsymbol{Q}(t)&=\begin{pmatrix}
\int_{t_1}^{t}H_{l_j+1}(\boldsymbol{s})dr  \\
 \vdots\\
\int_{t_1}^{t}H_n(\boldsymbol{s})dr
\end{pmatrix} \label{Q(t)} \ ,
    \end{align}
\end{subequations}
with $t_1=(k-1)T+\frac{T}{m}(j-1)$ and $t \leq t_2 = (k-1)T+\frac{T}{m}j$. One may note that ``$\epsilon$'' emerges in the integral of the second part of $\boldsymbol{P}(t)$, and we should point out that the term $\frac{h_{i,j}(\boldsymbol{s})-h_{i,j}(\boldsymbol{s}_0)}{\epsilon}$ can be converted into polynomial forms based on $\boldsymbol{P}$ and $\boldsymbol{Q}$ with elements of $\boldsymbol{s}$ and $\boldsymbol{s}_0$ as the coefficients. Focusing on $t=t_2$, the corresponding $\Delta(t)=(\boldsymbol{P}(t)\epsilon|\boldsymbol{Q}(t)\epsilon)$ actually describes the error between actual $\boldsymbol{s}(t)$ and ideal $\boldsymbol{s}_0(t)$ at the terminal time during the single segment $T^j_k$. Before we give further analysis on the expression of $\boldsymbol{P}(t)$ and $\boldsymbol{Q}(t)$, we set the following restrictions:
\begin{enumerate}
    \item $\dot{\boldsymbol{P}}(t), \dot{\boldsymbol{Q}}(t) \geq 0$ for $t_1 \leq t \leq t_2$, that is to say, errors accumulate unidirectionally in each segment $T^j_k$;
    \item $\boldsymbol{P}(t_1)=\boldsymbol{0}, \boldsymbol{Q}(t_1)=\boldsymbol{0}$ which means the accumulation of errors starts at the initial time of each segment $T^j_k$.
\end{enumerate}

The following lemma implies that we can linearize the expression of the second term of $\boldsymbol{P}(t)$, and then characterize its boundary. One can find the proof in Appendix.
\begin{lemma}\label{lemma:P(t)}
    The expression of the terminal error $\boldsymbol{P}(t_2)$ during the segment $T^j_k$ in \eqref{P(t)} can be further analyzed as
    \begin{subequations}
        \begin{align}
            (I_{l_j}-\rho  \hat{\Phi} ) \boldsymbol{P}(t) &\geq -\frac{H_0T}{m}  (\mathrm {1}_{l_j} + \rho \hat{\Psi} \cdot \mathrm {1}_{n-l_j}) \ , \\
            (I_{l_j}-\rho  \tilde{\Phi} ) \boldsymbol{P}(t) &\leq \frac{H_0T}{m} ( \mathrm {1}_{l_j} + \rho \tilde{\Psi} \cdot \mathrm {1}_{n-l_j}) \ ,
        \end{align}
    \end{subequations}
where $I_{l_j}$ is the identity matrix, $\rho=(x_h-p) \cdot \frac{T}{m}$ with parameters from our oscillator structure \eqref{detailed odes}, matrix $\hat{\Phi}_{l_j \times l_j}, \tilde{\Phi}_{l_j \times l_j}, \hat{\Psi}_{l_j \times (n-l_j)}$ and $\tilde{\Psi}_{l_j \times (n-l_j)}$ are decided by the polynomials $h_{1,j}, \dots, h_{l_j,j}$. $H_0$ is a upper bound of $|H_1(\boldsymbol{s})|, \dots, |H_{l_j}(\boldsymbol{s})|$ and each element of $\mathrm {1}_{l_j}$ and $\mathrm {1}_{n-l_j}$ is $1$.
\end{lemma}
Based on the lemma, we reach the conclusion about the boundary of the terminal error during each segment $T^j_k$. The proof and detailed expression are removed to Appendix.
\begin{theorem}\label{thm:error}
     $\boldsymbol{P}(t_2)$ and $\boldsymbol{Q}(t_2)$ are bounded with boundary controlled by the parameters from our oscillator structure under following three cases:
     \begin{enumerate}
         \item[\rm \rm 1.] $l_j=1$;
         \item[\rm \rm 2.] $(I_{l_j}-\rho \cdot \boldsymbol{H}_1)$ is positive definite for the linear system \eqref{regulation system} with $\boldsymbol{H}(\boldsymbol{s})=\boldsymbol{H} \cdot \boldsymbol{s}+ \boldsymbol{c}$ and matrix $\boldsymbol{H}_1$ is the $l_j$-dimentional leading principal submatrix;
         \item[\rm \rm 3.] both $(I_{l_j}-\rho \hat{\Phi})$ and $(I_{l_j}-\rho \tilde{\Phi})$ are positive definite as Lemma \ref{lemma:P(t)} shows.
     \end{enumerate}
In particular, both the upper bound of $\boldsymbol{P}(t_2)$ and $\boldsymbol{Q}(t_2)$ in the three cases are proportional to $\frac{H_0T}{m}$.
\end{theorem}

The upper boundary of $\boldsymbol{P}(t)$ we give is mainly determined by the kinetics of the modules to be regulated i.e. the matrix $\boldsymbol{H}(\boldsymbol{s})$, while the parameters $T$ and $\rho$ from our oscillator structure also have some influence. So far, we have tackled with the quantification of the error in a single segment $T^j_k$, and the accumulation of the error during the sequential execution of modules can be viewed as higher-order terms of $\epsilon$ appended to the original result which could be omitted under a sufficiently small $\epsilon$. We utilize an example of four-module regulation as follows to show the error during the first loop, separately $T^1_1, T^2_1, T^3_1$ and $T^4_1$.
\begin{example}\label{example:M4}
    Assume that we have four reaction modules to be regulated as follows:
    \begin{align*}
\mathcal{M}_1:
        S_{1} &\to S_{1} + S_{2}, &\mathcal{M}_2:  &S_{1} \to S_{1} + S_{2}, \\
        S_{2} &\to S_{1} + S_{2}, S_1 \to \varnothing, &   &S_{2} \to \varnothing, \\
        S_2 &\overset{2}{\rightarrow} \varnothing,  \varnothing \to S_2; &  \\
        \mathcal{M}_3: S_2 &\to S_1 + S_2, &\mathcal{M}_4:  &S_1 +S_2 \to S_1, \\
        S_1 &\to \varnothing;   & &\varnothing \overset{2}{\rightarrow} S_1, \varnothing \overset{2}{\rightarrow} S_2, \\
        & & &S_{1} \to \varnothing.
    \end{align*}
The whole kinetics express as
        \begin{align*}
        \dot{s}_1 &= (s_2-s_1)v_1 + (s_2-s_1)v_3 + (2-s_1)v_4\ ,\\
        \dot{s}_2 &= (s_1-2s_2+1)v_1 + (s_1-s_2+1)v_2 + (2-s_1s_2)v_4 
        \end{align*}
with $h_{1,1}=h_{1,3}=s_2-s_1, h_{1,2}=h_{2,3}=0, h_{1,4}=2-s_1, h_{2,1}=s_1-2s_2+1, h_{2,2}=s_1-s_2+1, h_{2,4}=2-s_1s_2$. 
\end{example}
Reactions in each module converge to the equilibrium exponentially, and the ideal output of the sequential execution of these four modules is $(1,1) \to (1,2) \to (2,2) \to (2,1)$. We use the oscillator model \eqref{detailed odes} to generate the four signals $v_1, v_2, v_3$ and $v_4$ as Fig. \ref{fig2} shows, with  $T=32.927, x_h-p=2, \rho=16.464$. We then demonstrate the error in every segment $T^j_1$ of the first loop ($j=1,2,3,4$). \\
\textbf{During $T^1_1$}: This is the case \rm 2 in Theorem \ref{thm:error}, with $\boldsymbol{\omega}=\left \{s_1,s_2\right \}$, $H_1=s_2-2s_1+2, H_2=s_1-s_2-s_1s_2+3$ and $\boldsymbol{H}_1=\boldsymbol{H}=\begin{pmatrix}
  -1&1 \\
  1&-2
\end{pmatrix}, (I_2-\rho \boldsymbol{H})^{-1}=\begin{pmatrix}
  0.1055&0.0512 \\
  0.0512&0.0543
\end{pmatrix}$. We can just choose $H_0=2$ and then $\boldsymbol{P}(t) \in \mathrm{R}^2$ satisfies $|P_1(t_2)| \leq 2.580, |P_2(t_2)| \leq 1.737$ with $t_2= 8.232$. 

\textbf{During $T^2_1$}: This is the case \rm 1 in Theorem \ref{thm:error}, with $\boldsymbol{\omega}=\left \{s_2\right \}$, $H_1=2s_2-3s_1+2, H_2=s_1-2s_2-s_1s_2+3$. This time $\boldsymbol{P}(t)$ becomes a scalar and $\boldsymbol{Q}(t)$ exists. We choose $H_0=3$, then $|Q(t_2)| \leq 24.695$, $\phi=-1$, $\psi$ corresponds to the upper boundary of $Q(t)$. So the inequality \eqref{xi} reduces to $|P(t_2)| \leq (\frac{H_0T}{m}+\rho \frac{H_0T}{m})/(1+\rho) = 24.695$ with $t_2=16.464$.

\textbf{During $T^3_1$}: This segment is similar to the previous one, with $\boldsymbol{\omega}=\left \{s_1\right \}$, $H_1=s_2-2s_1+2, H_2=2s_1-3s_2-s_1s_2+4$. We also choose $H_0=3$ and have $|P(t_2)|=|Q(t_2)| \leq \frac{H_0T}{m}= 24.695$ with $t_2= 24.695$.

\textbf{During $T^4_1$}: This is the case \rm 3 in Theorem \ref{thm:error}, with $\boldsymbol{\omega}=\left \{s_1, s_2\right \}$, $H_1=2s_2-2s_1, H_2=2s_1-3s_2+2$. We choose $H_0=3$, $\hat{\Phi}=\tilde{\Phi}=\begin{pmatrix}
  -1&0 \\
  -1&-2
\end{pmatrix}$, and $(I_2-\rho \hat{\Phi})^{-1}=(I_2-\rho \tilde{\Phi})^{-1}=\begin{pmatrix}
  0.0573&0 \\
  -0.0295&-0.0295
\end{pmatrix}$. Then we have $|P_1(t_2)| \leq 1.415$ and $P_2(t_2)=0$ with $t_2=32.927$.

\begin{figure}
\begin{center}
\includegraphics[width=1.0\linewidth,scale=1.00]{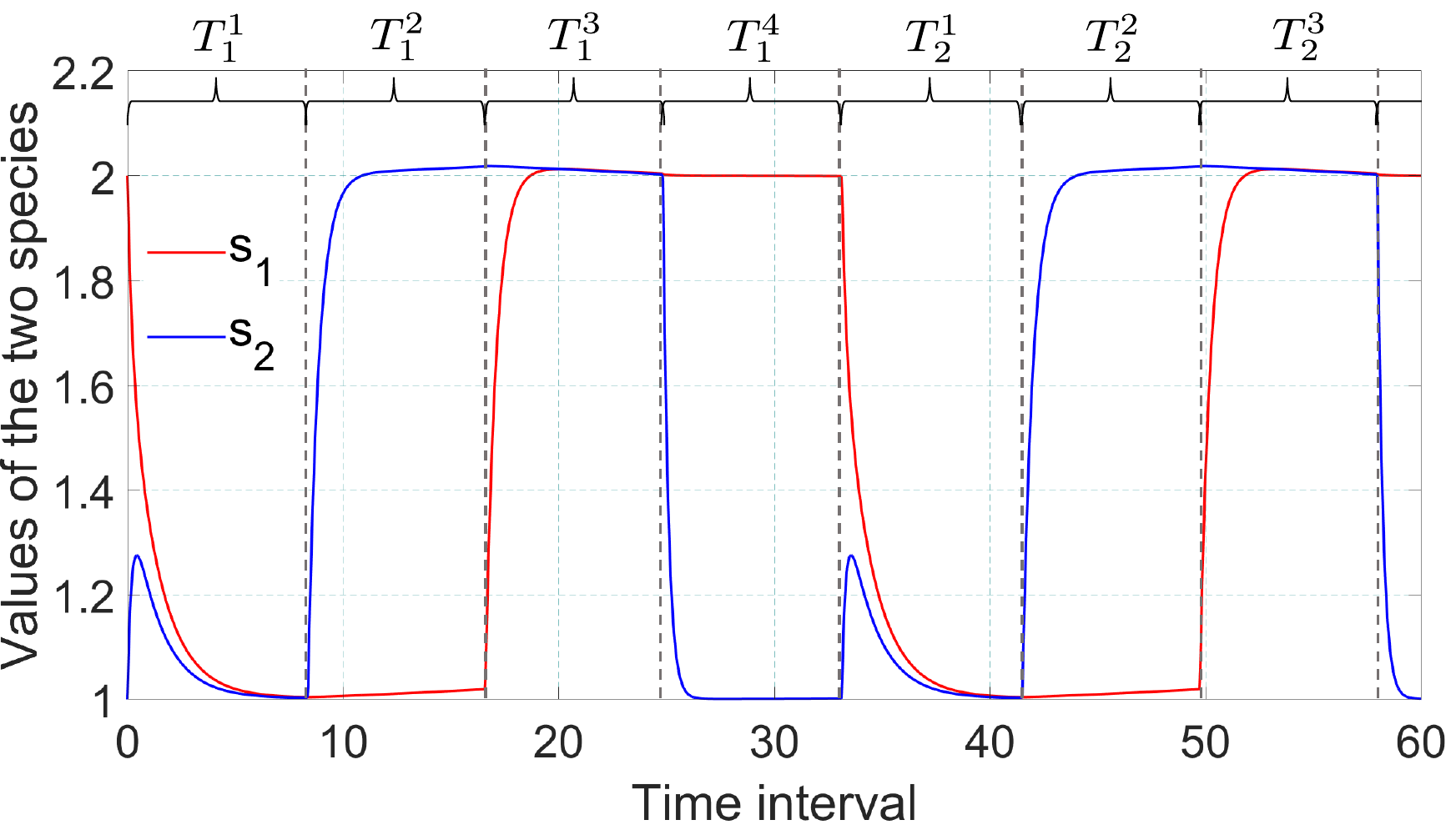}    % The printed column  
\caption{Simulation result of the Example \ref{example:M4}. Values of $s_1$ and $s_2$ are oscillating between $1$ and $2$, which is consistent with our expectation. }  % width is 8.4 cm.
\label{fig4}                                 % Size the figures 
\end{center}                                 % accordingly.
\end{figure}

The whole simulation result is provided in Fig. \ref{fig4} with $\epsilon_1=\epsilon_2=0.001$, $\epsilon=O(\epsilon_1)<0.002$, values of $s_1$ and $s_2$ are generally controlled between $1$ and $2$, just as we expect. Specially, the errors during $T^2_1$ and $T^3_1$ are larger than those during $T^1_1$ and $T^4_1$, which is also coincident with our analysis of upper boundary of $\boldsymbol{P}(t_2)$ and $\boldsymbol{Q}(t_2)$. The dynamic performance in the second loop is basically the same as that in the first loop, which also confirms that the errors accumulated between adjacent modules can be ignored. Noting that the period $T$ plays an important part in the errors of the whole regulation task, we demonstrate the simulation result of $s_1$ and $s_2$ with different choices of $T$ in Fig. \ref{fig5}. 
\begin{remark}
    Although it seems that a smaller $T$ induces smaller errors, when $T$ is insufficient to cover the time for the reactions to reach the equilibrium, the error caused by treating concentration of $S_1$ and $S_2$ in finite time as equilibrium emerges. Therefore, in the actual regulation, setting the whole period $T$ of the oscillator model \eqref{basic odes} needs to take both the complete execution time of the reactions in each module and the regulation error introduced by the oscillatory signals $V_j$ into account.
\end{remark}

\begin{figure}
\begin{center}
\includegraphics[width=1.0\linewidth,scale=1.00]{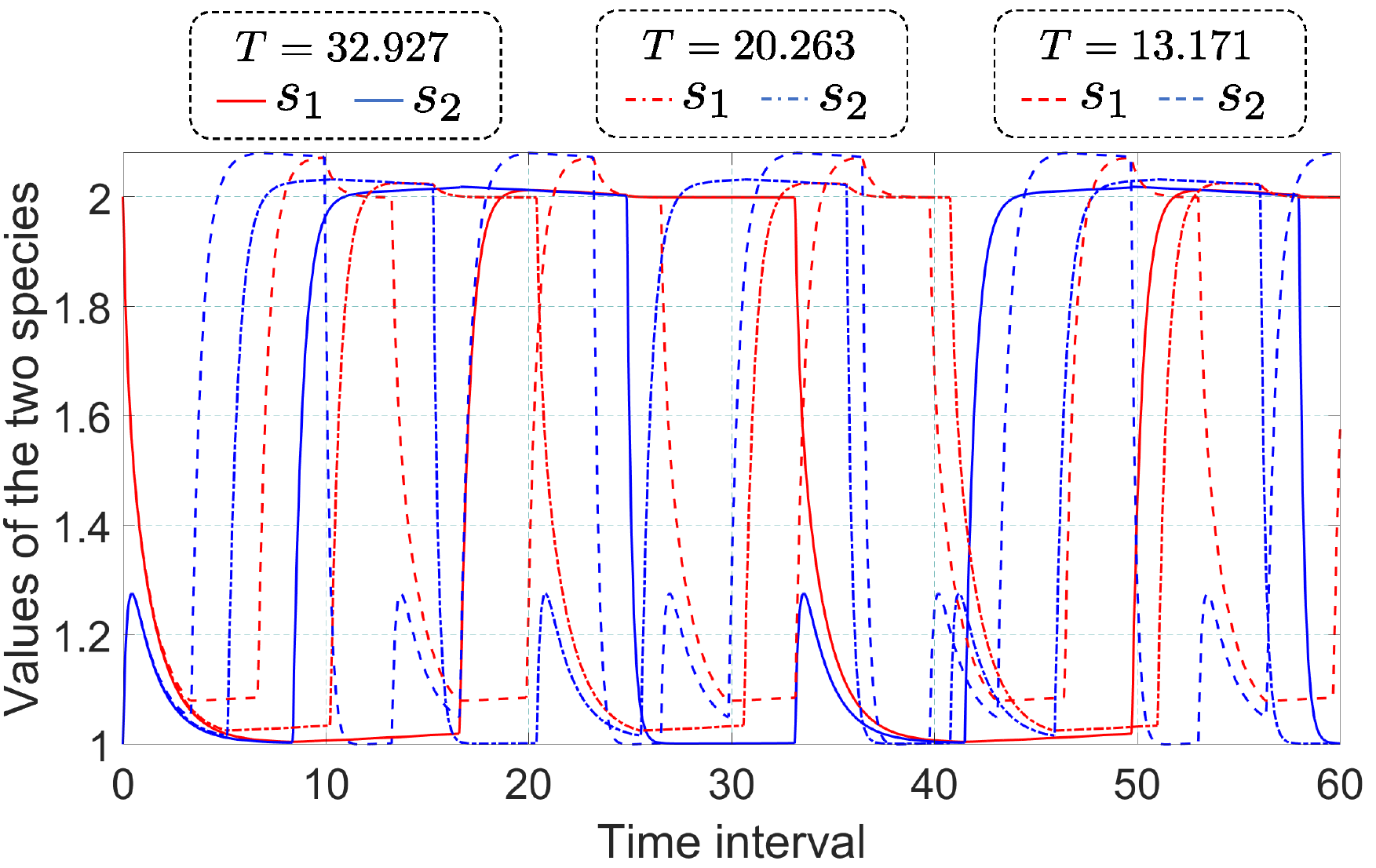}    % The printed column  
\caption{Simulation result of values of $s_1$ and $s_2$ based on different choices of $T$ in Example \ref{example:M4}. The values of $s_1$ and $s_2$ with the same period are marked in red and blue separately, and the values with different periods are distinguished by lines of different shapes.}  % width is 8.4 cm.
\label{fig5}                                 % Size the figures 
\end{center}                                 % accordingly.
\end{figure}

\section{Termination strategy and practical biochemical application}\label{sec:5}
In this section, we provide the general framework for the loop termination strategy and apply our design scheme for multiple periodical signals under a synthetic genetic oscillator and the Oregonator model, showing that our oscillator design provides guidance for module regulation tasks in practical biochemical experiments.

With respect to the computer instructions for loop termination, there are mainly two approaches: termination after a specified number of module loops and termination after the concentration of a species reaching a specific value. Whatever the case, we utilize a truncated subtraction module as follows to implement the conditional judgment which should be placed in front of all the modules to be regulated in actual tasks:
\begin{align*}\mathcal{M}_0:~
    &W_1 + W_0 \overset{\kappa}{\rightarrow} W_1 +2W_0 \ , \\
    &W_2 + W_0 \overset{\kappa}{\rightarrow} W_2 \ , 
    2W_0 \overset{\kappa}{\rightarrow} W_0 \ .
\end{align*}
The concentrations of $W_1$ and $W_2$ represent two variables in the conditional judgment i.e., the preset and current times of loops or the target and real-time concentrations of specific species. Species $W_0$ acts as a universal catalyst for the reactions in the subsequent modules, whose concentration converges exponentially to the truncated result of $w_1-w_2$, then the entire regulation system will be shut down when $w_1$ equals $w_2$. We can also control the reaction rate $\kappa$ to be lager than other reactions so that the loop termination can be done almost instantaneously. 
%\begin{figure*}
%\begin{center}
%\includegraphics[width=0.8\linewidth,scale=1.00]{fig6.pdf}    % The printed column  
%\caption{A schematic diagram for arbitrary regulation task of $m$ reaction modules. The modules in blue are target modules and the one in orange is our additional termination module. We use the dashed boxes to represent species and the cycle connected by the square arrows refers to the actual module regulation and loop.}  % width is 8.4 cm.
%\label{fig6}                                 % Size the figures 
%\end{center}                                 % accordingly.
%\end{figure*}

Another issue is the biochemical realization and application of our oscillator design. Usually, the modules to be regulated are described in the form of CRNs, and our oscillator model with general cubic form of $\varphi$ in the system \eqref{detailed odes} can be transformed back to CRNs as 
\begin{equation}\label{crn1}
    \begin{aligned}
   \xymatrix{ 3X \ar  @{ -^{>}}^{b\eta_1/\epsilon_1} @< 1pt> [r] & 4X  \ar  @{ -^{>}}^{\eta_1/\epsilon_1}  @< 1pt> [l] }\ , ~~~
   &\xymatrix{ X \ar  @{ -^{>}}^{c\eta_1/\epsilon_1} @< 1pt> [r] & 2X  \ar  @{ -^{>}}^{d\eta_1/\epsilon_1}  @< 1pt> [l] }\ , \\
X+2Y \xleftarrow{\eta_1} X+&Y \xrightarrow{\eta_1/\epsilon_1} Y \xrightarrow{\ell \eta_1} \varnothing\ ; \\
P \xrightarrow{\eta_1/\epsilon_2} P+U\ , ~~
&X \xrightarrow{\eta_1/\epsilon_2} X+V \ , \\
U \xrightarrow{\eta_1/\epsilon_2} \varnothing \ , ~ V \xrightarrow{\eta_1/\epsilon_2} &\varnothing\ , ~ U+V \xrightarrow{\eta_1/(\epsilon_1\epsilon_2)} \varnothing\ .
 \end{aligned}
\end{equation}
This CRN can be viewed as a specific biochemical oscillator which generates a clock signal $V$ under one selection of initial values. Technologies such as the DNA chain displacement cascades given by \cite{Soloveichik2010} can offer help to transform these abstract chemical reactions into real DNA chain reactions for experiments, and the transformation diagrams can be found in \citep{Shi2022}. In this paper we provide a new perspective. Noting that we actually demonstrate a design scheme for arbitrary number of symmetrical clock signals through our relaxation oscillator model \eqref{detailed odes}, the system \eqref{basic odes} can be replaced with any structure capable of producing relaxation oscillations. Our scheme is essentially to process a single relaxation oscillator into multiple periodic signals that meet our requirements in Definition \ref{def1}, which is illustrated in the following two examples.
\begin{example}\label{example:M}
    \cite{McMillen2002} proposed a synthetic genetic network in {\it Escherichia coli} with five fast reactions involving the binding of proteins:
    \begin{equation*}
        \begin{aligned}
            \xymatrix{ 4X \ar  @{ -^{>}}^{K_1} @< 1pt> [r] & X_4  \ar  @{ -^{>}}^{}  @< 1pt> [l] }, & ~
            \xymatrix{ D+ 4X \ar  @{ -^{>}}^{~~~K_2} @< 1pt> [r] & D_X \ar  @{ -^{>}}^{}  @< 1pt> [l] }, & \! \! \! \!
            \xymatrix{ A+R \ar  @{ -^{>}}^{~~~K_3} @< 1pt> [r] & C \ar  @{ -^{>}}^{}  @< 1pt> [l] }, \\
            \xymatrix{ C+C\ar  @{ -^{>}}^{~~~K_4} @< 1pt> [r] & C_2  \ar  @{ -^{>}}^{}  @< 1pt> [l] }, & ~
            \xymatrix{ D^L+ C_2 \ar  @{ -^{>}}^{~~~K_5} @< 1pt> [r] & D^L_C \ar  @{ -^{>}}^{}  @< 1pt> [l] }; \\
        \end{aligned}
    \end{equation*}
and six slow reactions involving transcription of mRNA and translation of proteins:
    \begin{equation*}
        \begin{aligned}
    D+P &\xrightarrow{k_t} D+P+n_xX+n_yY+n_lL\ , \\
    D_X+P &\xrightarrow{\alpha k_t} D_X+P+n_xX+n_yY+n_lL\ , \\
    D^L+P &\xrightarrow{k^L_t} D^L+P+n_xX\ , \\
    D^L_C+P &\xrightarrow{\beta k^L_t} D^L+P+n_xX\ , \\
    L+[substrates] &\xrightarrow{k_{la}} L+A\ , ~~~~ X+Y \xrightarrow{k_{xy}} Y\ .\\
        \end{aligned}
    \end{equation*}
The ODEs related to the relaxation oscillator are expressed in the following form through {\it Michaelis-Menten kinetics} after substitutions:
    \begin{subequations}\label{eq:ex5.1}
        \begin{align}
            \dot{x}&=m_xf(x)-\gamma_{xy}xy-\gamma_xx+\mu_xg(a) \ , \\
            \dot{y}&=m_yf(x)-\gamma_yy \ ,
        \end{align}
    \end{subequations}
    where $f(x)=(1+\alpha x^4)/(1+x^4)$, $g(a)=10$ at high levels of $a$. Both $x$ and $y$ refer to the concentrations of specific proteins. We follow the selection of parameters of \citep{McMillen2002} except for adjusting $m_y$ to 1.8, and apply the scheme of constructing four periodical signals to this model, the simulation result is shown in Fig. \ref{fig7}.
\end{example}

\begin{figure}
\begin{center}
\includegraphics[width=1.0\linewidth,scale=1.00]{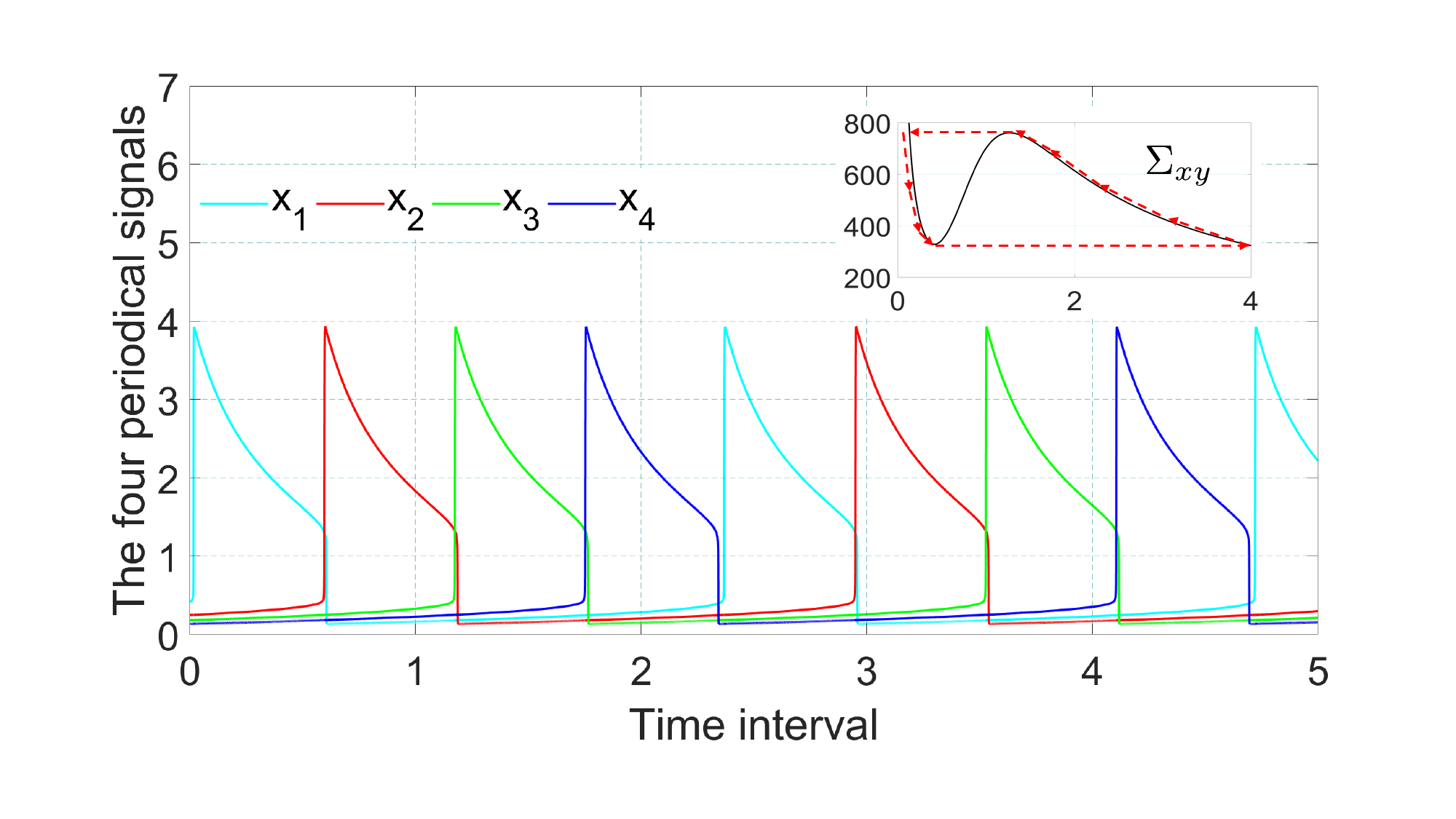}    % The printed column  
\caption{Design of four periodical signals under the oscillator model of \cite{Soloveichik2010}. The cubic-shaped critical manifold of $\Sigma_{xy}$ is displayed in the upper right corner, with the oscillation orbit in red dashed lines.}  % width is 8.4 cm.
\label{fig7}                                 % Size the figures 
\end{center}                                 % accordingly.
\end{figure}

\begin{example}\label{example:M2}
    \cite{Field1972} reduced the mechanism of the Belousov-Zhabotinskii reactions and gave the following Oregonator model consisting of five reactions:
\begin{equation*}
        \begin{aligned}
    BrO^-_3+Br^- &\xrightarrow{k_1} HBrO_2+HOBr\ ,  \\
    HBrO_2+Br^- &\xrightarrow{k_2} 2HOBr\ , \\
    BrO^-_3+HBrO_2 &\xrightarrow{k_3} 2HBrO_2+2Ce^{4+}\ ,  \\
    2HBrO_2 &\xrightarrow{k_4} BrO^-_3+HOBr\ ,\\
    Ce^{4+}+CH_2(COOH)_2&+BrCH(COOH)_2 \xrightarrow{k_5} hBr^-\ ,
        \end{aligned}
    \end{equation*}
where $h$ is a stoichiometric parameter. The two-variable Oregonator is expressed after substitutions as:
    \begin{subequations}
        \begin{align}
            \epsilon\dot{x}&=x(1-x)-fy\frac{x-q}{x+q} \ , \\
            \dot{y}&=kx-y \ ,
        \end{align}
    \end{subequations}
where $x$ and $y$ represent the concentration of $HBrO_2$ and $Ce^{4+}$, respectively, and $\epsilon$ is a sufficiently small constant, while $f, q$ and $k$ are parameters characterizing the reaction sequences. We follow the parameter selection of \citep{Field1972} except for adjusting $k$ to $0.74$, and also construct four clock signals based on this oscillator, the simulation result is provided in Fig. \ref{fig8}.
\end{example}

We obtain the diagrams in Fig. \ref{fig7} and Fig. \ref{fig8} similar to Fig. \ref{fig2}, which means that our design scheme can be transferred to other relaxation oscillation structures with more biochemical counterparts. Specifically, the result based on Example \ref{example:M} can be naturally applied to the sequence regulation of biochemical modules based on genetic regulatory networks, and Example \ref{example:M2} implies that our scheme could be combined with real chemical oscillations.
One may also notice that there is still some distance between the low amplitude of $x_i$ and $0$ in Fig. \ref{fig7}, which indicates that the design of our truncated subtraction module corresponding to the subsystem $\Sigma_{uv}$ in \eqref{detailed odes} is necessary. Therefore, our design scheme could provide guidance for module regulation tasks in practical biochemical experiments.

\begin{figure}
\begin{center}
\includegraphics[width=1.0\linewidth,scale=1.00]{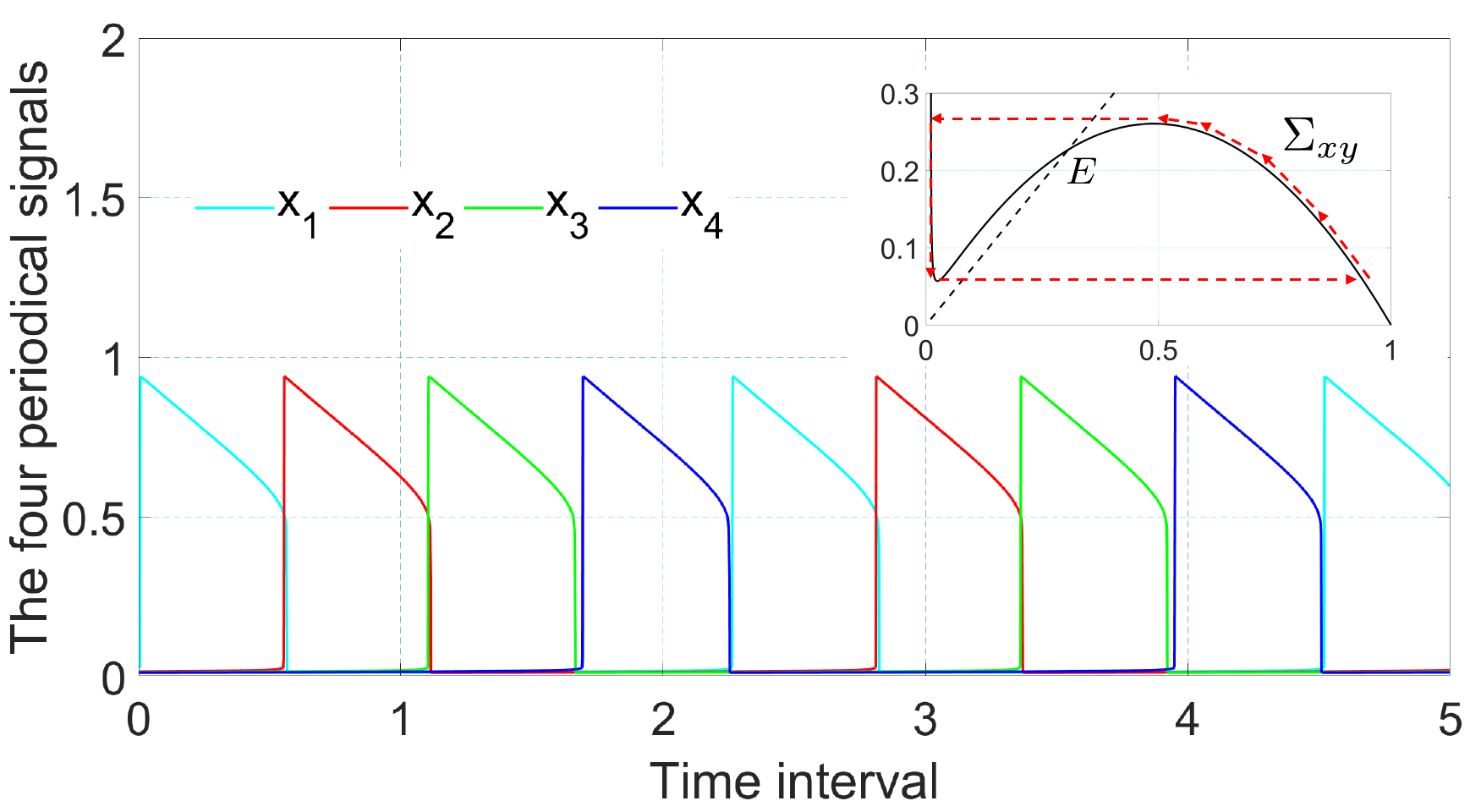}    % The printed column  
\caption{Design of four periodical signals under the Oregonator model given by \cite{Soloveichik2010}. The cubic-shaped critical manifold of $\Sigma_{xy}$ is displayed in the upper right corner, with the unstable equilibrium point $E$ and the oscillation orbit in red dashed lines.}  % width is 8.4 cm.
\label{fig8}                                 % Size the figures 
\end{center}                                 % accordingly.
\end{figure}
\section{Conclusion}\label{sec:6}
In this paper we show a universal scheme for arbitrary multi-module regulation and emphasize that the low amplitude of the oscillatory signals cannot be strictly $0$ (in fact the low amplitude of the biochemical oscillators to be $0$ is unimaginable). We then analyze the corresponding error induced by the oscillatory signals and provide the approach to control the error based on its structure. We also discuss the termination strategy and show the biochemical application of the whole design scheme.

In our future work, we will try to design suitable controller to replace the truncated subtraction module $\Sigma_{uv}$ in \eqref{detailed odes} in order to ``pull" the low amplitude of arbitrary relaxation oscillator $x$ like the one in \eqref{eq:ex5.1} to $0$, and combine more practical biochemical mechanisms and motifs with our oscillator model so that our design scheme could be capable of dealing with more specific biochemical systems and have broader application prospects.

\begin{ack}                               % Place acknowledgements
This work was funded by the National Nature Science Foundation of China under Grant No. 12320101001 and 12071428.  % here.
\end{ack}

\section*{Appendix: Proofs}
\textbf{Proof of Theorem \ref{thm1}.} We can directly integrate the formula as
\begin{align*}
    R(b,\alpha,\beta)&=\frac{-\frac{\sqrt{\alpha}}{2}Aln\frac{\frac{1}{36}b^2+\alpha}{\alpha}+F-\sqrt{\alpha}Cln\frac{\frac{1}{3}b+\beta}{\frac{2}{3}b+\beta}}{\frac{\sqrt{\alpha}}{2}Aln\frac{\frac{1}{9}b^2+\alpha}{\frac{1}{4}b^2+\alpha}+G-\sqrt{\alpha}Cln\frac{\frac{1}{6}b+\beta}{\beta}}, \\
    with \ A&=C=-\frac{\frac{1}{6}b+\beta}{\frac{1}{4}b^2+b\beta+\beta^2+\alpha}, \\
    B&=-\frac{(\frac{1}{3}b+\beta)(\frac{1}{6}b+\beta)}{\frac{1}{4}b^2+b\beta+\beta^2+\alpha}+1, \\
    F&=(\frac{1}{6}Ab+B)\arctan{\frac{\frac{1}{6}b}{\sqrt{\alpha}}}, \\
    G&=(\frac{1}{6}Ab+B)(\arctan{\frac{\frac{1}{3}b}{\sqrt{\alpha}}}-\arctan{\frac{\frac{1}{2}b}{\sqrt{\alpha}}}).\\
\end{align*}
Note that 
\begin{align*}
    &\lim_{\alpha,\beta \to 0}A=\lim_{\alpha,\beta \to 0}C=-\frac{2}{3b},\lim_{\alpha,\beta \to 0}B=\frac{7}{9}, \lim_{\alpha,\beta \to 0}F=\frac{\pi}{3}, \\
     &\lim_{\alpha,\beta \to 0}\sqrt{\alpha}Cln\frac{\frac{1}{3}b+\beta}{\frac{2}{3}b+\beta}= \lim_{\alpha \to 0} \frac{\sqrt{\alpha}}{2}Aln\frac{\frac{1}{9}b^2+\alpha}{\frac{1}{4}b^2+\alpha}=0 , \\
    &\lim_{\alpha \to 0}\frac{\sqrt{\alpha}}{2}Aln\frac{\frac{1}{36}b^2+\alpha}{\alpha}=\lim_{\beta \to 0}\sqrt{\beta}Cln\frac{\frac{1}{6}b+\beta}{\beta}=0, 
\end{align*}
$\lim_{\alpha,\beta \to 0}G=0$. As long as $\lim_{\alpha,\beta \to 0}\frac{\beta}{\alpha} \ne 0$, we have 
\begin{equation*}
\lim_{\alpha,\beta \to 0}R(b,\alpha ,\beta)=\infty, \ \forall b>0\ , 
\end{equation*}  
which reaches the conclusion.   $\hfill\blacksquare$

\textbf{Proof of Lemma \ref{lemma:P(t)}.} Since $h_{i,j}(\boldsymbol{s})$ refers to the kinetics of species $S_i$ in the module $\mathcal{M}_j$ in a polynomial form, a constant $H_0 >0 $ can be found such that $|H_{i}(\boldsymbol{s})| \leq H_0$ for bounded values of $\boldsymbol{s}$. It follows that $-\frac{H_0T}{m} \leq \int_0^tH_i(\boldsymbol{s})dr \leq \frac{H_0T}{m}$. Since $\boldsymbol{\omega}=\boldsymbol{\omega}_0+\boldsymbol{P}(t) \cdot \epsilon$, one can get $\frac{s_i-s_{0,i}}{\epsilon}=\boldsymbol{P}_i(t)$ for $\forall i =1, \dots, l_j$. Also the polynomial $h_{i,j}(\boldsymbol{s})-h_{i,j}(\boldsymbol{s}_0)$ can be rewritten in the form $\hat{h}(s_1-s_{0,1},\dots, s_{l_j}-s_{0,l_j}, \dots, s_n-s_{0,n})$, we could further reach the equation as follows:
     \begin{equation*}
         \frac{h_{i,j}(\boldsymbol{s})-h_{i,j}(\boldsymbol{s}_0)}{\epsilon} = \sum_{k=1}^{l_j} \Phi_{i,k} \boldsymbol{P}_k(t) +  \sum_{q=l_j+1}^{n} \Psi_{i,q} \boldsymbol{Q}_q(t)
     \end{equation*}
where both $\Phi_{i,k}$ and $\Psi_{i,q}$ consist of terms from $\boldsymbol{s}$ and $\boldsymbol{s}_0$ and can be regarded as the result of linearized reduction. Then we get
\begin{align*}
    &\int_{t_1}^t\frac{h_{i,j}(\boldsymbol{s})-h_{i,j}(\boldsymbol{s}_0)}{\epsilon}  v^j_j dr  \\
    =&(\sum_{k=1}^{l_j} \Phi_{i,k} \boldsymbol{P}_k(t_0) +  \sum_{q=l_j+1}^{n} \Psi_{i,q} \boldsymbol{Q}_q(t_0)) \cdot v^j_j(t_0) \cdot (t-t_1)
\end{align*}
with $t_1 \leq t_0 \leq t$. Because $\boldsymbol{s}$ and $\boldsymbol{s}_0$ are bounded and $v^j_j \leq x_h-p $, we can scale $\Phi_{i,k}$ and $\Psi_{i,k}$ to achieve $\hat{\Phi}_{i,k}$, $\tilde{\Phi}_{i,k}$, $\hat{\Psi}_{i,q}$ and $\tilde{\Psi}_{i,q}$ such that
    \begin{align*}
        &\int_{t_1}^t\frac{h_{i,j}(\boldsymbol{s})-h_{i,j}(\boldsymbol{s}_0)}{\epsilon}  v^j_j dr \\
       \geq& \rho \cdot (\sum_{k=1}^{l_j} \hat{\Phi}_{i,k} \cdot \boldsymbol{P}_k(t_2)+\sum_{q=l_j+1}^{n} \hat{\Psi}_{i,q} \cdot \boldsymbol{Q}_q(t_2)), \\
    &\int_{t_1}^t\frac{h_{i,j}(\boldsymbol{s})-h_{i,j}(\boldsymbol{s}_0)}{\epsilon}  v^j_j dr  \\
    \leq& \rho \cdot (\sum_{k=1}^{l_j} \tilde{\Phi}_{i,k} \cdot \boldsymbol{P}_k(t_2)+ \sum_{q=l_j+1}^{n} \tilde{\Psi}_{i,q} \cdot \boldsymbol{Q}_q(t_2).
    \end{align*}
Then it comes that
\begin{align*}
    \boldsymbol{P}(t_2) &\geq -\frac{H_0T}{m} \cdot \mathrm {1}_{l_j} + \rho \cdot (\hat{\Phi} \cdot \boldsymbol{P}(t_2) + \hat{\Psi} \cdot \boldsymbol{Q}(t_2))\  \\
    &\geq  -\frac{H_0T}{m} \cdot \mathrm {1}_{l_j} + \rho \hat{\Phi} \boldsymbol{P}(t_2) - \rho \frac{H_0T}{m} \hat{\Psi} \cdot \mathrm {1}_{n-l_j}\ , \\
     \boldsymbol{P}(t_2) &\leq \frac{H_0T}{m} \cdot \mathrm {1}_{l_j} + \rho \cdot (\tilde{\Phi} \cdot \boldsymbol{P}(t_2) + \tilde{\Psi} \cdot \boldsymbol{Q}(t_2)) \\
     &\leq  \frac{H_0T}{m} \cdot \mathrm {1}_{l_j} + \rho \hat{\Phi} \boldsymbol{P}(t_2) + \rho \frac{H_0T}{m} \hat{\Psi} \cdot \mathrm {1}_{n-l_j}\ ,
\end{align*}
and finally the inequalities emerge.     $\hfill\blacksquare$

\textbf{{Proof of Theorem \ref{thm:error}.}} Based on the boundary of $H_i$, it is obvious that $Q(t_2)$ is bounded, and we can just describe as 
   \begin{equation}
       -\frac{H_0T}{m} \cdot \mathrm {1}_{n-l_j} \leq \boldsymbol{Q}(t_2) \leq \frac{H_0T}{m} \cdot \mathrm {1}_{n-l_j}.
   \end{equation}
   In the first case $l_j=1$ means $\boldsymbol{P}(t)$ is actually a scalar, \eqref{P(t)} turns into one equation and one can easily get 
   \begin{equation}
      \rho \phi P(t_2) - \frac{H_0T}{m} (1+\rho \psi) \leq P(t_2) \leq \rho \phi P(t_2) + \frac{H_0T}{m} (1+\rho \psi)
   \end{equation}
where scalars $\phi$ and $\psi$ are both decided by the structure of $h_{1,j}(\boldsymbol{s})$ and separately, $\phi \cdot \frac{\partial h_{1,j}(\boldsymbol{s})}{\partial s_1}  \geq 0$, $\psi$ is the positive upper boundary of a specific polynomial function of $\boldsymbol{Q}(t_2)$. Since $\frac{\partial h_{1,j}(\boldsymbol{s})}{\partial s_1}  \leq 0$, we have $\phi \leq 0$ and finally 
\begin{equation}\label{xi}
     |P(t_2)| \leq \frac{H_0T}{m}(1+\rho \psi)/(1-\rho \phi) \ .
\end{equation}
    In the second case, similar to the proof of Lemma \ref{lemma:P(t)}, \eqref{P(t)} can be reduced into
    \begin{align*}
         \boldsymbol{P}(t_2) &\geq \rho \boldsymbol{H}_1 \cdot \boldsymbol{P}(t_2) - \frac{H_0T}{m}(\mathrm {1}_{l_j} + \rho \boldsymbol{H}_2 \cdot \mathrm {1}_{n-l_j}) \ , \\
       \boldsymbol{P}(t_2) &\leq \rho \boldsymbol{H}_1 \cdot \boldsymbol{P}(t_2) + \frac{H_0T}{m}(\mathrm {1}_{l_j} + \rho \boldsymbol{H}_2 \cdot \mathrm {1}_{n-l_j}) \ ,
    \end{align*}
   with the submatrix $\boldsymbol{H}_2$ determined by the rows from $1$ to $l_j$ and columns from $l_j+1$ to $n$. As long as $I_{l_j}-\rho \cdot \boldsymbol{H}$ is positive definite, $(I_{l_j}-\rho \cdot \boldsymbol{H})^{-1}$ exists and is also positive definite, then it follows
   \begin{align*}
       \boldsymbol{P}(t_2) &\geq -\frac{H_0T}{m} (I_{l_j}-\rho \boldsymbol{H}_1)^{-1} ( \mathrm {1}_{l_j}+ \rho \boldsymbol{H}_2 \cdot \mathrm {1}_{n-l_j}) \ , \\
       \boldsymbol{P}(t_2) &\leq \frac{H_0T}{m} (I_{l_j}-\rho \boldsymbol{H}_1)^{-1} ( \mathrm {1}_{l_j}+ \rho \boldsymbol{H}_2 \cdot \mathrm {1}_{n-l_j}) \ . 
   \end{align*}
  In the third case, benefit directly from Lemma \ref{lemma:P(t)}, we have
  \begin{align*}
       \boldsymbol{P}(t_2) &\geq -\frac{H_0T}{m} (I_{l_j}-\rho \hat{\Phi})^{-1} ( \mathrm {1}_{l_j}+ \rho \hat{\Psi} \cdot \mathrm {1}_{n-l_j}) \ , \\
       \boldsymbol{P}(t_2) &\leq \frac{H_0T}{m} (I_{l_j}-\rho \tilde{\Phi})^{-1} ( \mathrm {1}_{l_j}+ \rho \tilde{\Psi} \cdot \mathrm {1}_{n-l_j}) \ ,
  \end{align*}
  with $\hat{\Phi}$, $\hat{\Psi}$, $\tilde{\Phi}$ and $\tilde{\Psi}$ come from the process of linearized reduction.  $\hfill\blacksquare$

\begin{wrapfigure}{l}{0mm}
\includegraphics[width=0.95in,height=1.25in,clip,keepaspectratio]{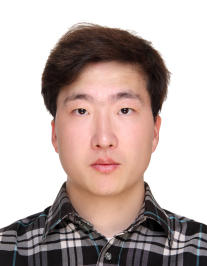}
\end{wrapfigure}
Xiaopeng Shi received the B.Sc.degree in Applied Mathematics from Wuhan University, China, in 2020, and has studied as a Ph.D. candidate in Operational Research and Cybernetics of Zhejiang University from September 2020 till now.
His research interests are in the areas of chemical reaction network theory, oscillation in dynamical system and control theory.

\begin{wrapfigure}{l}{0mm}
     \includegraphics[width=1in,height=1.25in,clip,keepaspectratio]{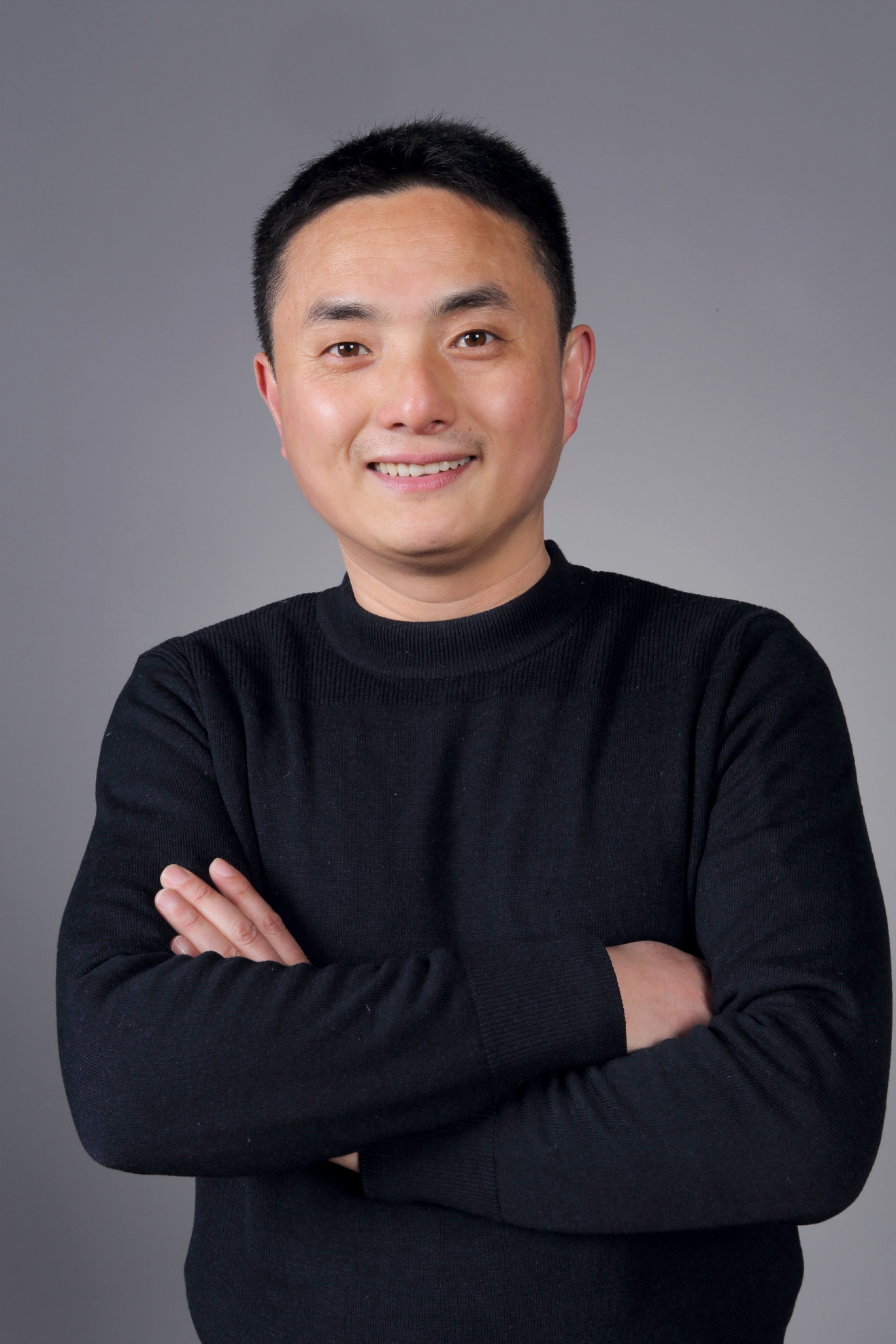}
\end{wrapfigure}
 Chuanhou Gao received the B.Sc. degrees in Chemical Engineering from Zhejiang University of Technology, China, in 1998, and the Ph.D. degrees in Operational Research and Cybernetics from Zhejiang University, China, in 2004. From June 2004 until May 2006, he was a Postdoctor in the Department of Control Science and Engineering at Zhejiang University. Since June 2006, he has joined the Department of Mathematics at Zhejiang University, where he is currently a Professor. He was a visiting scholar at Carnegie Mellon University from Oct. 2011 to Oct. 2012. 
 His research interests are in the areas of data-driven modeling, control and optimization, chemical reaction network theory and thermodynamic process control. He is an associate editor of IEEE Transactions on Automatic Control and of International Journal of Adaptive Control and Signal Processing.

\begin{wrapfigure}{l}{0mm}
	\includegraphics[width=1in,height=1.25in,clip,keepaspectratio]{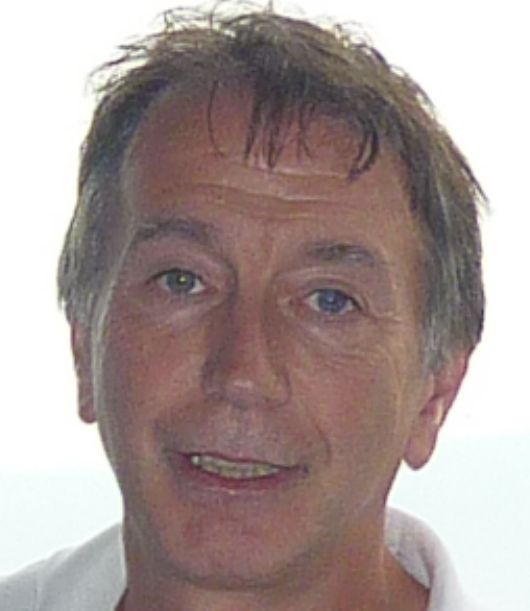}
\end{wrapfigure}
Denis Dochain received his degree in Electrical engineering in 1982 from the Université Catholique de Louvain, Belgium. He completed his Ph.D. thesis and a ``th\`{e}se d'agr\'{e}gation de l'enseignement sup\'{e}rieur" in 1986 and 1994, respectively, also at the Université Catholique de Louvain, Belgium. He has been CNRS associate researcher at the LAAS (Toulouse, France) in 1989, and Professor at the Ecole Polytechnique de Montr\'{e}al, Canada in 1987-88 and 1990-92. He has been with the FNRS (Fonds National de la Recherche Scientifique, National Fund for Scientific Research), Belgium since 1990. Since September 1999, he is Professor at the ICTEAM (Institute), Universit\'{e} Catholique de Louvain, Belgium, and Honorary Research Director of the FNRS. He has been invited professor at Queen's University, Kingston, Canada between 2002 and 2004. He is full professor at the UCL since 2005. He is the Editor-in-Chief of the Journal of Process Control, senior editor of the IEEE Transactions of Automatic Control and associate editor of Automatica. He is active in IFAC since 1999 (Council member, Technical Board member, Publication Committee member and chair, TC and CC chair). He received the IFAC outstanding service award in 2008 and is an IFAC fellow since 2010. He received the title of Doctor Honoris Causa from the INP Grenoble on December 13, 2020.
His main research interests are in the field of nonlinear systems, thermodynamics based control, parameter and state estimation, adaptive extremum seeking control and distributed parameter systems, with application to microbial ecology, environmental, biological and chemical systems, and electrical and mechanical systems. He is the (co-)author of 5 books, more than 160 papers in refereed journals and 260 international conference papers.
\end{document}